\documentclass[11pt]{article}
\usepackage[letterpaper, margin=1in]{geometry}
\usepackage{amsfonts,amsmath,amssymb,amsthm}
\usepackage{cite,comment,enumerate,hyperref}
\usepackage[ruled,vlined,linesnumbered]{algorithm2e}
\usepackage{thmtools}
\usepackage{thm-restate}
\usepackage[utf8]{inputenc}
\usepackage[T1]{fontenc}

\newtheorem{theorem}{Theorem}
\newtheorem{corollary}[theorem]{Corollary}
\newtheorem{observation}{Observation}

\newtheorem{lemma}[theorem]{Lemma}
\newtheorem{claim}{Claim}
\theoremstyle{claim}
\theoremstyle{definition}
\newtheorem{definition}{Definition}
\theoremstyle{remark}	

\usepackage{todonotes}

\title{Near-Optimal Decremental Hopsets with Applications}
\date{}
\author{Jakub Łącki \\
Google Research\\
jlacki@google.com
\and Yasamin Nazari \\
University of Salzburg\footnote{This work was conducted in part while the author was an intern at Google and a PhD student at Johns Hopkins University. Supported in part by NSF award CCF-1909111 and by Austrian Science Fund (FWF) grant P 32863-N.}\\
ynazari@cs.sbg.ac.at
}

\begin{document}
\maketitle
\thispagestyle{empty}

\begin{abstract}
     Given a weighted undirected graph $G=(V,E,w)$, a hopset $H$ of \emph{hopbound} $\beta$ and \emph{stretch} $(1+\epsilon)$ is a set of edges such that for any pair of nodes $u, v \in V$, there is a path in $G \cup H$ of at most $\beta$ hops, whose length is within a $(1+\epsilon)$ factor from the distance between $u$ and $v$ in $G$. 
     We show the first efficient decremental algorithm for maintaining hopsets with a \textit{polylogarithmic} hopbound.
     The update time of our algorithm matches the best known static algorithm up to polylogarithmic factors. All the previous decremental hopset constructions had a \textit{superpolylogarithmic} (but subpolynomial) hopbound of $2^{\log^{\Omega(1)} n}$ [Bernstein, FOCS'09; HKN, FOCS'14; Chechik, FOCS'18].

     By applying our decremental hopset construction, we get improved or near optimal bounds for several distance problems.
     Most importantly, we show how to decrementally maintain $(2k-1)(1+\epsilon)$-approximate all-pairs shortest paths (for any constant $k \geq 2)$, in $\tilde{O}(n^{1/k})$ amortized update time\footnote{Throughout this paper we use the notation $\tilde{O}(f(n))$ to hide factors of $O(\text{polylog } (f(n)))$.}  and $O(k)$ query time.
     This improves (by a polynomial factor) over the update-time of the best previously known decremental algorithm in the \textit{constant} query time regime. Moreover, it improves over the result of [Chechik, FOCS'18] that has a query time of $O(\log \log(nW))$, where $W$ is the aspect ratio, and the amortized update time is $n^{1/k}\cdot(\frac{1}{\epsilon})^{\tilde{O}(\sqrt{\log n})})$. For sparse graphs our construction nearly matches the best known static running time / query time tradeoff. 
    
     We also obtain near-optimal bounds for maintaining approximate multi-source shortest paths and distance sketches, and get improved bounds for approximate single-source shortest paths. Our algorithms are randomized and our bounds hold with high probability against an \textit{oblivious} adversary.

\end{abstract}

\section{Introduction}

Given a weighted undirected graph $G=(V,E,w)$, a hopset $H$ of \emph{hopbound} $\beta$ and \emph{stretch} $(1+\epsilon)$ (or, a $(\beta, 1+\epsilon)$-hopset) is a set of edges such that for any pair of nodes $u, v \in V$, there is a path in $G \cup H$ of at most $\beta$ hops, whose length is within a $(1+\epsilon)$ factor from the distance between $u$ and $v$ in $G$ (see Definition~\ref{def:hopset} for a formal statement).

Hopsets, originally defined by \cite{cohen2000}, are widely used in distance related problems in various settings, such as parallel shortest path computation \cite{cohen2000,miller2015,elkin2019almost, elkin2019RNC}, distributed shortest path computation \cite{elkin2019journal, nanongkai2014, censor2019}, routing tables \cite{elkin2017}, and distance sketches \cite{elkin2017, dinitz2019}.
In addition to their direct applications, hopsets have recently gained more attention as a fundamental object (e.g.~\cite{merav2020, elkin2019journal,abboud2018,huang2019}), and are known to be closely related to several other fundamental objects such as additive (or near-additive) spanners and emulators~\cite{elkin2020survey}.

A key parameter of a hopset is its hopbound. In many settings, after constructing a hopset, we can approximate distances in a time that is proportional to the hopbound. For instance, in parallel or distributed settings a hopset with a hopbound of $\beta$ allows us to compute approximate single-source shortest path in $\beta$ parallel rounds (e.g. by using Bellman-Ford). 
For many applications, such as approximate APSP (all-pairs shortest paths), MSSP (multi-source shortest paths), computing distance sketches, and diameter approximation, where we require computing distances from \textit{many} sources, we are interested in the regime where the hopbound is polylogairthmic. Indeed, we obtain improved (and in some cases near-optimal) bounds for several of these problems in decremental settings.

In this paper, we study the maintenance of hopsets in a dynamic setting.
Namely, we give an algorithm that given a weighted undirected graph $G$ maintains a hopset of $G$ under edge deletions.
Our algorithm covers a wide range of hopbound/update time/hopset size tradeoffs. Importantly, we get the first efficient algorithm for decrementally maintaining a hopset with a \textit{polylogarithmic} hopbound.
In this case, assuming $G$ initially has $m$ edges and $n$ vertices, our algorithm takes $O(mn^{\rho})$ time, given any constant $\rho > 0$, and maintains a hopset of polylogarithmic hopbound and $1+\epsilon$ stretch.
This matches (up to polylogarithmic factors) the 
running time of the best known static algorithm~\cite{elkin2019RNC, elkin2019journal} for computing a hopset with polylogarithimic hopbound and $(1+\epsilon)$ stretch.

\begin{theorem}
Given an undirected graph $G=(V,E)$ with polynomial weights\footnote{If weights are not polynomial the $\log n$ factor will be replaced with $\log W$ in the hopbound, and a factor of $\log^2 W$ will be added to the update time, where $W$ is the aspect ratio (the ratio between largest and smallest distance).}, subject to edge deletions, we can maintain a $(\beta, 1+\epsilon)$-hopset of size $\tilde{O}(n^{1+\frac{1}{2^k -1}})$ in total expected update time  $\tilde{O}(\frac{\beta}{\epsilon} \cdot (m+n^{1+\frac{1}{2^k -1}})n^{\rho})$, where $\beta= (O(\frac{\log n}{\epsilon} \cdot (k+1/\rho)))^{k+1/\rho+1}$, $k \geq 1$ is an integer, $0 < \epsilon <1$ and $\frac{2}{2^k-1} < \rho <1$. 
\end{theorem}

In the decremental setting, to the best of our knowledge, the previous state-of-the art hopset constructions have a hopbound of $2^{\tilde{O}(\log^{3/4} n)}$~\cite{henzinger2014}, or
$(1/\epsilon)^{\tilde{O}(\sqrt{\log n})}$~\cite{bernstein2009,chechik2018}. As a special case, by setting $\rho=(2^k-1)^{-1}=\frac{\log \log n}{\sqrt{\log n}}$, we can maintain a hopset with hopbound $2^{\tilde{O}(\sqrt{\log n})}$ in $2^{\tilde{O}(\sqrt{\log n})}$ amortized time.
More importantly, by setting $\rho$ and $k$ to a constant, we can maintain a hopset of \textit{polylogarithmic} hopbound. 

While hopsets are extensively studied in other models of computation (e.g.~distributed and parallel settings), their applicability in dynamic settings is less understood.
Examples of results utilizing hopsets include the state-of-the art decremental SSSP algorithm for undirected graphs by Henzinger, Krinninger and Nanongkai~\cite{henzinger2014}, and implicit hopsets considered in \cite{bernstein2009, chechik2018}. As stated, these decremental hopset algorithms as stated only provide a \emph{superpolylogarithmic} hopbound. It may be possible (while not discussed) to use the hop-reduction techniques of \cite{henzinger2014} (inspired by a similar technique in \cite{bernstein2009}) to obtain a wider-range of tradeoffs, however to the best of our knowledge these techniques do not lead to near-optimal size/hopbound tradeoffs\footnote{ In particular, in all regimes the algorithm of \cite{henzinger2014} gives a hopset with size that is super-linear in number of edges $m$  (e.g. $m^{1+p}$ for a parameter $p$), while our hopset size is $O(n^{1+p})$ for some (other but similar) parameter $p$, which is a constant when the hopbound is polylogarithmic.  Moreover, our techniques lead to near-optimal approximate APSP, whereas it is unclear how to get comparable bounds using techniques in \cite{henzinger2014}, as they do not maintain Thorup Zwick-based clusters.}. 
Hence our result constitutes the first near-optimal decremental algorithm for maintaining hopsets with in a wide-range of settings including polylogarithmic hopbound.

\paragraph{Discussion on hopset limitations and alternative techniques.}

In \cite{abboud2018} it was shown that for a $(\beta, 1+\epsilon)$-hopset with size $n^{1+\frac{1}{{2^k}-1}- \delta}$ for any fixed $k, \epsilon$ and $\delta >0$ we must have\footnote{$\Omega_k$ hides exponential a factor of roughly $1/(k2^k)$. As written in \cite{abboud2018} they assume $k$ is constant (and hence the sparse hopset regime is not covered), but they also indicate that a tighter analysis could change the exact relationship between $k$ and $\epsilon$ and hence allow a better $k$ dependence and covering the sparse case (see Theorem 4.6 and Remark 4.7 in \cite{abboud2018}).} $\beta = \Omega_k(\frac{1}{\epsilon})^k$. Their lower bound suggests that we cannot construct a $(\beta, 1+\epsilon)$-hopset of size $\tilde{O}(n)$ with $\beta = \textrm{poly} \log(n)$ hopbound, implying that hopsets cannot be used for obtaining optimal time (i.e.~polylog amortized time) for \textit{sparse} graphs and \textit{very small} $\epsilon$.
However when the graph is slightly denser ($|E|= n^{1+\Omega(1)}$), the approximation factor is slightly larger (see e.g.~\cite{merav2020, elkin2019almost}), or we aim to compute distances from many sources (in APSP or MSSP), using hopsets may still lead to optimal algorithms. Indeed, we show that our decremental hopsets allow us to obtain a running time matching the best static algorithm (up to polylogarithmic factors) both in $(2k-1)$-APSP and $(1+\epsilon)$-MSSP. We leave it as an open problem if hopsets can be used to obtain linear time algorithms for SSSP with \textit{larger} approximation factors (e.g.~$\epsilon \geq 1$), since as stated, the lower bound of \cite{abboud2018} does not apply in this case. 

It is worth noting that in Theorem \ref{thm:restricted_hopset} we first give a decremental algorithm that maintains static hopsets of \cite{elkin2019RNC} that matches the size/hopbound tradeoff in the lower bound of \cite{abboud2018}. However, as we will see, this algorithm has a large update time, and thus we propose a new hopset with slightly worse size/hopbound tradeoff that can be maintained much more efficiently. This efficient variant has additional polylogarithmic (in aspect ratio) factors in the hopbound relative to the existentially optimal construction.

Finally, for \textit{single source} shortest path computation in other models recently algorithms based on continuous optimization techniques are proposed (e.g. \cite{andoni2020,becker2021,li2020}) that outperform algorithms based only on combinatorial objects such as hopsets/emulators. These optimization techniques lead to much better dependence on $\epsilon$, but are less suitable when there are many sources, as the running time scales with the number of sources. Interestingly, the authors of \cite{andoni2020} use low-hop combinatorial structures with \textit{larger (polylogrithmic)} stretch as a subroutine in their continuous optimization framework. Hence understanding both combinatorial and optimization directions seems crucial for distance computation in general. 

\subsection{Applications of Our Decremental Hopsets}

To illustrate applicability of our decremental hopset algorithm, we show how it yields improved algorithms for decremetanlly maintaining shortest paths from a fixed set $S$ of sources.
We consider different variants of the problem which differ in the size of $S$: the single-source shortest paths (SSSP) problem ($|S|=1$), all-pairs shortest paths (APSP) problem ($S = n$, where $n$ is the number of vertices of the input graph), as well as the multi-source shortest paths (MSSP) problem ($S$ is of arbitrary size), which is a generalization of the previous two.

\paragraph{Near-Optimal approximate APSP.}
We give a new decremental algorithm for maintaining approximate all-pairs shortest paths (APSP) with \text{constant} query time.

\begin{theorem}[Approximate APSP]
For any constant integer\footnote{The $k$ here should not be confused with the parameter $k$ in the hopset size.} $k \geq 2$, there is a data structure that can answer $(2k-1)(1+\epsilon)$-approximate distance queries in a given a weighted undirected graph $G=(V, E, w)$ subject to edge deletions.
The total expected update time over any sequence of edge deletions is $\tilde{O}(mn^{1/k})$ and the expected size of the data structure is $\tilde{O}(m+n^{1+1/k})$.
Each query for the distance between two vertices is answered in $O(k)$ worst-case time.
\end{theorem}

Our result improves upon a decremental APSP algorithm by Chechik~\cite{chechik2018} in a twofold way.
First, for constant $k$, our update time bound is better by a $(1/\epsilon)^{{O}(\sqrt{\log n})}$ factor.
Second, we bring down the query time from $O(\log \log (n W))$ to constant.
We note that in the area of distance oracles
a major goal is to preprocess a data structure that can return a distance estimate in \textit{constant} time \cite{mendel2007, wulff2012, roditty2005det, chechik2014}\footnote{We need to store the original graph in addition to the distance oracle in order to update the distances and maintain correctness, however we do \textit{not} need the whole graph for \textit{querying distances} as we will also point out in describing the applications in maintaining distance sketches.}.

Our results match the best known static algorithm with the same tradeoff (up to $(1+\epsilon)$ in the stretch and polylog in time) by Thorup-Zwick \cite{TZ2005} for sparse graphs. For dense graphs there have been improvements by \cite{wulff2012} in static settings. 

Prior to \cite{chechik2018}, Roditty and Zwick \cite{roditty2004} gave an algorithm for maintaining Thorup-Zwick distance oracles in total time $\tilde{O}(mn)$, stretch $(2k-1)(1+\epsilon)$ and $O(k)$ query time for \textit{unweighted graphs}. Later on, Bernstein and Roditty \cite{bernstein2011} gave a decremental algorithm for maintaining Thorup-Zwick distance oracles in $O(n^{2+1/k+o(1)})$ time using emulators also only for \textit{unweighted graphs}.
\paragraph{Distance Sketches.} 
Another application of our hopsets with polylogarithmic hopbound is a near-optimal decremental algorithm for maintaining distance sketches (or distance labeling); an important tool in the context of distance computation. The goal is to store a small amount of information, a sketch, for each node, such that the distance between any pair of nodes can be approximated only using their sketches (without accessing the rest of the graph). Distance sketches are particularly important in networks, and distributed systems \cite{sarma2015, elkin2017}, and large-scale graph processing \cite{dinitz2019}. Their significance is that at query time we only need to access/communicate the small sketches rather than having to access the whole graph. This is specially useful for processing large data when queries happen more frequently than updates. 

The Thorup-Zwick \cite{TZ2005} algorithm can be used to obtain distance sketches of expected size $O(kn^{1/k})$ (for each node) that supports $(2k-1)$-approximate queries in $O(k)$ time (in static settings), and this is known to be tight assuming a well-known girth conjecture. Our approximate APSP data structure has the additional property that the information stored for each node is a distance sketch of expected size $O(kn^{1/k})$ that supports $(2k-1)(1+\epsilon)$-approximate queries. Hence we can maintain distance sketches that almost match the guarantees of the best static algorithm. More specifically, for a fixed size our algorithm matches the best known static construction up to a $(1+\epsilon)$-factor in the stretch and polyloagrithmic factors in the update time. In decremental settings, distance oracles of \cite{TZ2005}, and hence distance sketches with the guarantees described are studied by \cite{roditty2004, bernstein2011}, but our total update time of $\tilde{O}(mn^{1/k})$ (for constant $k \geq 2$) significantly improves over these results. In particular \cite{roditty2004} maintains these distance sketches in a total update time of $\Omega(mn)$, and \cite{bernstein2011} requires total update time of $O(n^{2+1/k+o(1)})$.

\paragraph{Near-Optimal $(1+\epsilon)$-MSSP.} Our next result is a near-optimal algorithm for multi-source shortest paths.

\begin{theorem}[MSSP]
There is a data structure which given a weighted undirected graph $G=(V, E)$ \emph{explicitly} maintains $(1+\epsilon)$-approximate distances from a set of $s$ sources in $G$ under edge deletions, where $0<\epsilon<\frac{1}{2}$ is a constant.
Assuming that $|E|= n^{1+\Omega(1)}$ and $s=n^{\Omega(1)}$, the total expected update time is $\tilde{O}(sm)$.
The data structure is randomized and works against an oblivious adversary.
\end{theorem}

We note that total update time matches (up to polylogarithmic factors) the running time of the best known \emph{static} algorithm for computing $(1+\epsilon)$-approximate distances from $s$ sources for a wide range of graph densities.
While for very dense graphs, using algorithms based on fast matrix multiplication is faster, the running time of our decremental algorithm matches the best known results in the static settings (up to polylogarithmic factors) whenever $ms = n^\delta$, for a constant $\delta \in (1, 2.37)$.

In the dynamic setting, our algorithm improves upon algorithms obtained by using hopsets of Henzinger, Krinninger and Nanongkai~\cite{henzinger2014}, or emulators of Chechik~\cite{chechik2018}, both of which give a total update time of $O(sm \cdot 2^{\tilde{O}(\log^{\gamma} n)}), 0<\gamma <1$ (for a constant $\gamma$).
In particular, by maintaining a hopset with polylogarithmic hopbound in $\tilde{O}(sm)$ time, we can maintain approximate SSSP from each source in $\tilde{O}(m)$ time. In contrast, in~\cite{henzinger2014, chechik2018} with hopset of hopbound $2^{\tilde{O}(\log^{\gamma} n)}$ is maintained, which if one simply applies existing techniques, results in a total update time of $m2^{\tilde{O}(\log^{\gamma} n)}$. 
In the general case, i.e., for very sparse graphs, the update bound of our algorithm is $sm{2^{\tilde{O}(\sqrt{\log n}}})$, which is similar but slightly better than the bound obtained by~\cite{henzinger2014}, and slightly improves over dependence on $\epsilon$ over \cite{chechik2018}.

\paragraph{Improved bounds for $(1+\epsilon)$-SSSP.} 
Finally, in order to better demonstrate how our techniques compare to previous work, we show that we can obtain a slightly improved bound for decremental single-source shortest paths.

\begin{theorem}
Given an undirected and weighted graph $G=(V, E)$, there is data structure for maintaining $(1+\epsilon)$-approximate distances from a source $s_0 \in V$ under edge deletions, where $0 <\epsilon<1$ is a constant and $|E|= n \cdot 2^{\tilde{\Omega}(\sqrt{\log n})}$. The total expected update time of the data structure is $m \cdot 2^{\tilde{O}(\sqrt{\log n})}$. There is an additional factor of $O(\frac{1}{\epsilon})^{\frac{\sqrt{\log n}}{\log \log n}}$ in the running time for non-constant $\epsilon$. 
\end{theorem}

The amortized update time of our algorithm over all $m$ deletions  is $2^{\tilde{O}(\sqrt{\log n})}$.
This improves upon the state-of-the art algorithm of \cite{henzinger2014}, whose amortized update time is $2^{\tilde{O}(\log ^{3/4} n)}$.
We note that the techniques of \cite{chechik2018} can also be used to obtain $(1+\epsilon)$-SSSP in amortized update time $\tilde{O}(1/\epsilon)^{\sqrt{\log n}}$. This is close to our update time, but we get a better bound with respect to the dependence on $\epsilon$. 

\paragraph{Recent developments on decremental shortest paths.} Recently and after a preprint of this paper was published, a decremental \textit{deterministic} $(1+\epsilon)$-SSSP also with amortized update time of $n^{o(1)}$ was proposed by \cite{bernstein2021deterministic}. Several other recent results have also focused on deterministic dynamic shortest path algorithms or algorithms that work against an \textit{adaptive adversary} (e.g.~\cite{bernstein2017, bernstein2017ICALP, gutenberg2020, bernstein2021deterministic}) most of which also use hopsets or related objects such as emulators. Our work leaves an open problem on whether hopsets with small hopbound can also be maintained and utilized deterministically\footnote{One possible direction is considering derandomization of Throup-Zwick based clustering in static settings\cite{TZ2005} combined with our techniques.}. This could have applications in deterministic approximate all-pairs shortest paths, which could in turn have implications in using decremental shortest path algorithms for obtaining faster algorithms in classic/static settings (e.g. see \cite{madry2010}).

\paragraph{Hopsets vs. emulators.}
A majority of the previous work on dynamic distance computation are based on sparse \textit{emulators} (e.g.~\cite{bernstein2009, bernstein2011, chechik2018}). For a graph $G=(V,E)$, an emulator $H'=(V, E')$ is a graph such that for any pair of nodes $x,y \in V$, there is a path in $H'$ that approximates the distance between $x$ and $y$ on $G$ (possibly with both multiplicative and additive factors). While there are some similarities in algorithms for constructing these objects, their analysis is different.
More importantly, their maintenance and utilization for dynamic shortest paths have significant differences. An emulator approximates distances without using the original graph edges and hence we can restrict the computation to a sparser graph, whereas for using hopsets we also need the edges in the original graph. On the other hand, hopsets allow one to only consider paths with few hops. 

\subsection{Preliminaries and Notation}
Given a weighted undirected graph $G=(V,E, w)$, and a pair $u,v \in V$ we denote the (weighted) shortest path distance by $d_G(u,v)$. We denote by $d_G^{(h)}(u,v)$ the length of the shortest path between $u$ and $v$ among the paths that use at most $h$-hops, and call this the $h$-hop limited distance between $u$ and $v$. 

In this paper, we are interested in designing \textit{decremental} algorithms for distance problems in weighted graphs. In the decremental setting, the  updates are only edge deletions or weight increases. This is as opposed to an \textit{incremental} setting in which edges can be inserted, or a \textit{fully dynamic} setting, in which we have both insertions and deletions. Specifically, given a weighted graph $G= (V,E, w)$, we want to support the following operations: \textsc{Delete}($(u,v)$), where $(u,v) \in E$, which removes the edge $(u,v)$, \textsc{Distance}($s,u$), which returns an (approximate) distance between a source $s$ and any $u \in V$, and \textsc{Increase}($(u,v), \delta$), which increases the weight of the edge $(u,v)$ by $\delta > 0$. While our results also allow handling weight increases, in stating our theorems for simplicity we use the term \textit{total update time} to refer to a sequence of up to $m$ deletions.

\begin{definition}\label{def:hopset}
Let $G = (V, E, w)$ be a weighted undirected graph.
Fix $d, \epsilon > 0$ and an integer $\beta \geq 1$.
A $(d, \beta, 1+\epsilon)$-\emph{hopset} is a graph $H = (V, E(H), w_H)$ such that for each $u, v \in V$, where $d_G(u, v) \leq d$, we have $d_G(u,v) \leq d_{G \cup H}^{(\beta)}(u, v) \leq (1+\epsilon) d_G(u, v)$.
We say that $\beta$ is the \emph{hopbound} of the hopset and $1+\epsilon$ is the \emph{stretch} of the hopset. 
We also use $(\beta, 1+\epsilon)$-hopset to denote a $(\infty, \beta, 1+\epsilon)$-hopset.
\end{definition}

We sometimes call a $(d, \beta, 1+\epsilon)$-hopset a $d$-\emph{restricted hopset}, when the other parameters are clear. We also sometimes consider hopset edges added for a specific distance range $(2^{j}, 2^{j+1}]$, which we call a hopset for a single distance \textit{scale}. We also use $W$ to denote the ratio between maximum and minimum weight in $G$, also called the aspect ratio. W.l.o.g we can assume that the maximum distance is bounded by $nW$.

In analyzing dynamic algorithms we sometimes also use a time subscript $t$ to denote a distance (or a weight) after the first $t$ updates. In particular we use $d_{t,G}(u,v)$ to denote the distance between $u$ and $v$ after $t$ updates, and similarly use
$d^{(h)}_{t,G}(u,v)$ to denote $h$-hop limited distance between $u$ and $v$ at time $t$.

\section{Overview of Our Algorithms}\label{sec:overview}
The starting point of our algorithm is a known static hopset construction \cite{elkin2019RNC, huang2019}. We first review this construction. As we shall see, maintaining this data structure dynamically directly would require update time of up to $O(mn)$. Our first technical contribution is another hopset construction that captures some of the properties of the hopsets of \cite{elkin2019RNC, huang2019}, but can be maintained efficiently in a decremental setting. We then explain how by hierarchically maintaining a sequence of data structures we can obtain a near-optimal time and stretch tradeoff. 

\subsection{Static Hopset of \cite{elkin2017}}\label{sec:static_hopset}
In this section we outline the (static) hopset construction of Elkin and Neiman~\cite{elkin2019RNC}\footnote{In \cite{elkin2019RNC} two algorithms with different sampling probabilities are given, where one removes a factor of $k$ in the size. This factor does not impact our overall running time, so we will use the simpler version.} (which is similar to \cite{huang2019}).
We will later give a new (static) hopset algorithm that utilizes some of the properties of this construction but with modifications that allows us to maintain a \textit{similar} hopset dynamically.


\begin{definition}[Bunches and clusters]\label{def:bunches}
Let $G=(V,E, w)$ be a weighted, $n$-vertex graph, $k$ be an integer such that $1 \leq k \leq \log \log n$ and $\rho > 0$.
We define sets $V=A_0 \supseteq A_1 \supseteq ... \supseteq A_{k+ 1/\rho+1}=\emptyset$. Let $\nu =\frac{1}{2^k-1}$. Each set $A_{i+1}$ is obtained by sampling each element from $A_i$ with probability $q_i=\max(n^{-2^i \cdot \nu}, n^{-\rho})$.

Fix $0 \leq i \leq k+1/\rho+1$ and for every vertex $u \in A_i\setminus A_{i+1}$, let $p(u) \in A_{i+1}$ be the node of $A_{i+1}$, which is closest to $u$, and let $d(u,A_{i+1}):=d(u,p(u))$ (assume $d(u, \emptyset)= \infty$). We call $p(u)$ the \emph{pivot} of $u$.
We define a \emph{bunch} of $u$ to be a set $B(u):= \{ v \in A_i: d(u,v) < d(u,A_{i+1})\}$. Also, we define a set $C(v)$, called the \emph{cluster} of $v \in A_i \setminus A_{i+1}$, defined as $C(v)=\{ u \in V: d(u,v) <d(u, A_{i+1}) \}$. 
\end{definition}

Note that if $v \in B(u)$ then $u \in C(v)$, but the converse does not necessarily hold.
The way we define the bunches and clusters here follows~\cite{elkin2019RNC}, but differs slightly from the definitions in~\cite{TZ2005, roditty2004}, where each vertex has a separate bunch and cluster defined for each level $i$ (and stores the union of these for all levels).

The clusters are \textit{connected} in a sense that if a node $u \in C(v)$ then any node $z$ on the shortest path between $v$ and $u$ is also in $C(v)$. This property is important for bounding the running time (as also noted in \cite{TZ2005, roditty2004}):

\begin{claim}
Let $u \in C(v)$, and let $z \in V$ be on a shortest path between $v$ and $u$. Then $z \in C(v)$.
\end{claim}
\begin{proof}
Let $v \in A_i$. If $z \not \in C(v)$ then by definition $d(z,A_{i+1}) \leq d(v,z)$. On the other hand, since  $z$ is 
on the shortest path between $u$ and $v$: $d(u,A_{i+1}) \leq d(z,u)+ d(z,A_{i+1}) \leq d(u,z)+d(z,v)=d(u,v)$, which contradicts the fact that $u \in C(v)$.
\end{proof}

The hopset is then obtained by adding an edge $(u,v)$ for each $u \in A_i \setminus A_{i+1}$ and $v \in  B(u)\cup \{p(u)\}$, and setting the weight of this edge to be $d(u,v)$.
These distances can be computed by maintaining the clusters.

\begin{lemma}[\hspace{1sp}\cite{elkin2019RNC,huang2019}]\label{lem:bunches}
Let $G=(V,E, w)$ be a weighted, $n$-vertex graph, $k$ be an integer such that $1 \leq k \leq \log \log n$ and $0< \rho, 0 < \epsilon <1$.
Assume the sets $A_i$ and bunches are defined as in Definition~\ref{def:bunches}.
Define a graph $H=(V, E_H, w_H)$, such that for each $u \in A_i \setminus A_{i+1}$ and $v \in  B(u)\cup \{p(u)\}$, we have an edge $(u, w) \in E_H$ with weight $d_G(u,v)$.
Then $H$ is a $(\beta, 1+\epsilon)$-hopset of size $O(n^{1+\frac{1}{2^k-1}})$, where $\beta= (O(\frac{k+1/\rho}{\epsilon})^{k+1/\rho+1}$.
\end{lemma}

For reference we sketch a proof of the hopset properties in Appendix \ref{app:static_hopset}. Our main result is based on a new construction consisted of a hierarchy of hopsets. Our dynamic hopset requires a new \emph{stretch} analysis as estimates on the shortest paths are obtained from different data structures, but the \emph{size} analysis is basically the same. 

While we are generally interested in a hopset that is not much denser than the input, as we will see the running time (both in static and dynamic settings) is mainly determined by the number of clusters a node belongs to, rather than the size of the hopset.
Moreover, unlike an emulator, for computing the distances using a hopset, we also need to consider the edges in $G$, and a small hopbound is the key to efficiency rather than the sparsity.

The hopset of \cite{elkin2019RNC} has some structural similarities to the emulators of \cite{TZ2006spanners}. One main difference is that the sampling probabilities are adjusted (lower-bounded by $n^{-\rho}$) to allow for efficient construction of these hopsets in various models, at the cost of slightly weaker size/hopbound tradeoffs. This adjustment is also crucial for our efficient decremental algorithms. Inspired by the construction described, in Section \ref{sec:new_static} we describe a \textit{new static} hopset algorithm, and later in Section \ref{sec:overview_dynamic} we adapt it to decremental settings.

\subsection{New static hopset based on path doubling and scaling}\label{sec:new_static}
As a warm-up, before moving to our new dynamic hopset construction, we provide a simple static hopset and explain why we expect to maintain such a structure more efficiently than the structure in Section \ref{sec:static_hopset} in \textit{dynamic settings}. Our main contribution is to maintain a dynamic hopset \textit{efficiently} using ideas in the simple algorithm described in this section.

At a high level, computing one of the main components of the hopset of Lemma~\ref{lem:bunches} involves multiple single-source shortest paths computations.
Maintaining single-source shortest paths is easy in the decremental setting, if we limit ourselves to paths of low length (or allow approximation).
Namely, assuming integer edge weights, one can maintain single source shortest paths up to length $d$ under edge deletions in total $O(md)$ time.

If we simply modified the construction of the hopset of Lemma~\ref{lem:bunches}, and computed shortest paths up to length $d$ instead of shortest paths of unbounded length, we would obtain a $d$-restricted hopset.
We describe this idea in more detail in Section \ref{sec:restricted_hopset}, where we show that an adaptation of the techniques by \cite{roditty2004} allows us to maintain a $d$-restricted hopset in deceremental settings, in total time $O(dmn^{\rho})$, for a parameter $0<\rho<\frac{1}{2}$. However, for large $d$ such a running time is prohibitive. In order to address this challenge, in this section we describe a static hopset, which can be computed using shortest path explorations up to only a polylogarithmic depth, yet can be used to approximate arbitrarily large distances. In the next sections we leverage this property to maintain a similar hopset in the decremental setting efficiently. This will require overcoming other obstacles, notably the fact that the dynamic shortest path problems that we need to solve are not decremental.
\paragraph{Path doubling.}
Assume that we are given a procedure $\textsc{Hopset}(G, \beta, d, \epsilon)$ that constructs a $(d, \beta, 1+\epsilon)$ hopset. In Section \ref{sec:new_hopset}, we provide such an algorithm that uses only shortest path computation up to polylogarithmic depth. We argue that by applying the $\textsc{Hopset}(G, \beta, d, \epsilon)$ procedure \textit{repeatedly} we can compute a full hopset, and in this process by utilizing the previously added hopset edges we can restrict our attention to short-hop paths only.

More formally, we construct a sequence of graphs $H_0, \ldots, H_{\log (nW)}$, such that $H_j$ is a hopset that handles pairs of nodes with distance in range $[2^{j-1}, 2^j)$, for $0 \leq j \leq \log (nW)$. This implies that $\bigcup_{r=0}^j H_r$ is a $(2^j, \beta, (1+\epsilon)^j)$-hopset of $G$.
Note that for $0 \leq j \leq \log \beta$ we can set $H_j = \emptyset$, since $G$ covers these scales. 
We would like to use $G \cup_{r=0}^{j-1} H_{r}$ to construct $H_j$ based on the following observation that has been previously used in other (static) models (e.g. parallel hopsets of Cohen \cite{cohen2000}).

Consider $u,v, d_G(u,v) \in [2^{j-1}, 2^j)$, and let $\pi$ be the shortest path between $u$ and $v$ in $G$. Then $\pi$ can be divided into three segments $\pi_1, \pi_2$ and $\pi_3$, where $\pi_1$ and $\pi_3$ have length at most $2^{j-1}$ and $\pi_2$ consists of a single edge. But we know there is a path in $G \cup_{r=0}^{j-1} H_{r}$ with at most $\beta$-hops that approximates each of $\pi_1$ and $\pi_2$. Hence for constructing $H_j$ we can compute \textit{approximate} shortest paths by restricting our attention to paths of consisting of at most $2\beta+1$ \textit{hops} in $G \cup_{r=0}^{j-1} H_{r}$.

This idea, which we call path doubling, has been previously used in hopset constructions in distributed/parallel models (e.g.~\cite{cohen2000, elkin2019journal, elkin2019RNC}), but to the best of our knowledge this is the first use of this approach in a dynamic setting. Applying this idea in parallel/distributed settings is relatively straight-forward, since having bounded hop paths already leads to efficient parallel shortest path explorations (e.g. by using Bellman-Ford). However, utilizing it efficiently in dynamic settings is more involved for several reasons: we have to simultaneously maintain the clusters (including their connectivity property), apply a scaling idea on the whole structure, and handle insertions in our hopset algorithm and its analysis.

But first we describe a \textit{scaling idea} that at a high-level allows to go from \textit{$h$-hop-bounded} explorations to $h$-depth bounded explorations on a scaled graph.

\paragraph{Scaling.} 
 We review a scaling algorithm that allows us to utilize the path doubling idea. Similar scaling techniques are used in dynamic settings~\cite{bernstein2009, bernstein2011, henzinger2014} for single-source shortest paths, but as we will see, using the scaling idea in our setting is more involved since it has to be carefully combined with other components of our construction.

This idea can summarized in the following scaling scheme due to Klein and Subramanian \cite{klein1997}, which, roughly speaking, says that finding shortest paths of length $\in [2^{j-1}, 2^j)$ and at most $\ell$ hops, can be (approximately) reduced to finding paths of length at most $O(\ell)$ in a graph with in integral weights. This is done by a rounding procedure that adds a small \textit{additive} term of roughly $\frac{\epsilon_0 w(e)}{\ell}$ to each edge $e$. Then for a path of $\ell$ hops the overall stretch will be $(1+\epsilon_0)$.

\begin{lemma}[\cite{klein1997}]\label{lem:rounding}
Let $G=(V,E,w)$ be a weighted undirected graph. Let $R \geq 0$ and  $\ell \geq 1$ be integers and $\epsilon_0 > 0$.
We define the \emph{scaled graph} to be a graph $\textsc{Scale}(G, R, \epsilon_0, \ell) :=  (V, E, \hat{w})$, such that $\hat{w}(e) = \lceil \frac{w(e)}{\eta(R, \epsilon_0)} \rceil$, where $\eta(R, \epsilon_0)=\frac{\epsilon_0 R}{\ell}$. Then for any path $\pi$ in $G$ such that $\pi$ has at most $\ell$ hops and weight $R \leq w(\pi) \leq 2R$ we have, 
\begin{itemize}
    \item $\hat{w}(\pi) \leq \lceil 2\ell/\epsilon_0 \rceil$,
   \item $w(\pi) \leq \eta(R, \epsilon_0) \cdot \hat{w}(\pi) \leq (1+\epsilon_0) w(\pi)$.
\end{itemize}
\end{lemma}

Similar scaling ideas have been used in the $h$-SSSP algorithm for maintaining approximate shortest paths~\cite{bernstein2009}.
The algorithm maintains a collection of trees and to return a distance estimate, it finds the tree that best approximates a given distance. But we note that in utilizing the scaling techniques in our final dynamic hopset construction we cannot simply maintain a \textit{disjoint set of} bounded hop shortest path trees. We need to maintain the whole structure of the hopset on the scaled graphs together: firstly, based on definition of bunches in Lemma \ref{lem:bunches}, nodes keep on leaving and joining clusters, so we cannot simply maintain a set of shortest trees from a fixed set of roots. We need to maintain the \textit{connectivity} of clusters as described in Section \ref{sec:static_hopset} at the same time as maintaining the shortest path trees. Additionally, while we are maintaining distances over the set of clusters we also need to handle insertions introduced by smaller scales.

To maintain these efficiently, we need to apply the scaling to the whole structure, including the hopset edges added so far. But when we utilize the smaller scale hopset edges (for applying path doubling) insertions or distance decreases are introduced. As we will see, handling insertions at the same time the clusters (and the corresponding distances) are updated makes the stretch/hopbound analysis more involved.

Next we combine the scaling with the path doubling techniques. The path doubling property states that we can restrict our attention to $(2\beta+1)$-hop limited shortest path computation, and the scaling idea ensures that such $(2\beta+1)$-\textit{hop} bounded paths in $G \cup \bigcup_{r=0}^{j-1} H_r$ correspond to paths bounded \textit{in depth} by $d=\lceil \frac{2\ell}{\epsilon_0} \rceil= O(\frac{\beta}{\epsilon_0})$ in the scaled graph $G_{scaled}=\textsc{Scale}(G \cup \bigcup_{r=0}^{j-1} H_r, 2^j, \epsilon_0, 2\beta+1)$. Informally, this mean it is enough to construct shortest path trees up to depth $d$ on the scaled graphs in our hopset construction.

 We can now summarize our new static hopset construction in Algorithm \ref{alg:static_hopset}. Similarly to $\textsc{Scale}$, for a graph $G = (V, E, w)$ we define $\textsc{Unscale}(G, R, \epsilon, \ell)$ to be a graph $(V, E, w')$, where for each $e \in E$, $w'(e) = \eta(R, \epsilon) \cdot w(e)$. In static settings, the procedure $\textsc{Hopset}(G, \beta, d,\epsilon)$ for constructing a $d$-restricted hopset can be performed by running a $(\beta,1+\epsilon)$-hopset construction algorithm in which the shortest path explorations are restricted to depth $d$. In Section \ref{sec:restricted_hopset} we describe a decremental algorithm for this procedure, and describe how it leads to a $(d,\beta,1+\epsilon)$-restricted hopset. Note that we can set $\beta = \textrm{poly} \log n$, and so the shortest path explorations can be bounded by a polylogairthmic value.

\begin{algorithm}[H]
\label{alg:static_hopset}
\SetKwProg{Fn}{Function}{}{}
 
\For{$j=1$ to $\lceil \log W \rceil$}{
$G_{scaled} := \textsc{Scale}(G \cup \bigcup_{r=0}^{j-1} H_r, 2^j, \epsilon_0, 2\beta+1)$\\
$\hat{H}:= \textsc{Hopset}(G_{scaled}, \beta, \lceil \frac{2(2\beta+1)}{\epsilon_0}\rceil, \epsilon) $\\
$H_j := \textsc{Unscale}(\hat{H}, 2^j, \epsilon_0, 2\beta+1))$\\
$H:= H \cup H_j$
 }
\caption{}
\end{algorithm}

It is not hard to see that in such a static construction, three different approximation factors are combined in each scale: a $(1+\epsilon)$-stretch due to the \textsc{Hopset} procedure, a $(1+\epsilon_1)$-factor from the restricted hopset, and a $(1+\epsilon_0)$-factor due to scaling. This is summarized in the following lemma.

\begin{lemma}
Let $G$ be a graph and $H$ be a $(d, \beta, 1+\epsilon_1)$ hopset of $G$.
Let $G_{scaled} = \textsc{Scale}(G \cup H, d, \epsilon_0, 2\beta+1)$
and let $H' = \textsc{Unscale}(\textsc{Hopset}(G_{scaled}, d, \lceil \frac{2(2\beta+1)}{\epsilon_0}\rceil, \epsilon_2))$.
Then $H \cup H'$ is a $(2d, \beta, (1+\epsilon)(1+\epsilon_1)(1+\epsilon_0))$ hopset of $G$.
\end{lemma}

Obtaining such a guarantee in dynamic settings is going to be more involved, since we also need to handle insertions, and at the same time ensure that the update time remains small. Moreover the stretch analysis will require combining estimates obtained by different procedures.

\subsection{Near-Optimal Decremental Hopsets}\label{sec:overview_dynamic}

In this section we describe how we can overcome the challenges of the dynamic settings in order to maintain a decremental hopset in near-optimal update-time.

The first step of our algorithm is constructing a $d$-restricted version of the hopset described in Section \ref{sec:static_hopset}. As discussed, for this we can use the techniques by \cite{roditty2004} to maintain a $(d, \beta, 1+\epsilon)$-hopset in $\tilde{O}(d m n^{\rho})$ total update time, where $0<\rho<\frac{1}{2}$. Now in order to remove the time dependence on $d$, we use the path doubling and scaling ideas described as follows: we maintain this data structure on a sequence of scaled graphs simultaneously, and argue that this data structure gives us a hopset on $G$ after \textit{unscaling} the edge weights.


 

\paragraph{Sequence of restricted hopsets.} Similiar to Section \ref{sec:new_hopset}, our decremtnal algorithm maintains the sequence of graphs $H_0, \ldots, H_{\log W}$, where for each $0 \leq j \leq \log W$, $\bigcup_{r=0}^j H_r$ is a $(2^j, \beta, (1+\epsilon)^j)$-hopset of $G$. For \textit{each scale} we show the following:


\begin{restatable}{lemma}{singlescale}\label{lem:single-scale}
Consider a graph $G = (V, E, w)$ subject to edge deletions, and parameters $0 < \epsilon<1, \rho <\frac{1}{2}$.
Assume that we have maintained $\bar{H}_j:=H_1,...,H_j$, which is a $(2^j, \beta, (1+\epsilon)^{j})$-hopset of $G$. Then given the sequence of changes to $G$ and $\bar{H}_j$, we can maintain a graph $H_{j+1}$, such that $\bar{H}_j \cup H_{j+1}$ is a $(2^{j+1}, \beta, (1+\epsilon)^{j+1})$-hopset of $G$. This restricted hopset can be maintained in $\tilde{O}( (m+\Delta) n^{\rho} \cdot \frac{\beta}{\epsilon})$ total time, where $m$ is the initial size of $G$, and $\Delta$ is the number of edges inserted to $\bar{H}_j$ over all updates, $\beta= ( \frac{1}{\epsilon \cdot \rho})^{O(1/\rho)}$.
\end{restatable}
Note that the lemma does not hold for \textit{any} restricted hopset, and in dynamic settings we need to use special properties of our construction to prove this. 

To prove this lemma we use the techniques of \cite{roditty2004} to maintain the clusters. For obtaining near-optimal update time, we combine this algorithm with the path doubling and scaling ideas described earlier. However, in order to utilize these ideas, we need to deal with the fact that inserting hopset edges from smaller scales introduces \textit{insertions}.

\paragraph{Handling insertions.}  In addition to maintaining clusters and distances decrementally, in our final construction we need to handle edge \emph{insertions}.
This is because we run it on a graph $G \cup \bigcup_{r=0}^{j-1} H_{r}$ (after applying scaling of Lemma~\ref{lem:rounding}).
While edges of $G$ can only be deleted, new edges are added to the $H$ that we need to take into account for obtaining faster algorithms.

At a high-level, the algorithm of Roditty and Zwick~\cite{roditty2004} decrementally maintains a collection of single-source shortest path trees (up to a bounded depth) using the Even-Shiloach algorithm (ES-tree)~\cite{ES} at the same time as maintaining a clustering.
To handle edge insertions, we modify the algorithm to use the \emph{monotone} ES-tree idea proposed by~\cite{henzinger2014, henzinger2016}.

The goal of a monotone ES-tree is to support edge insertions in a limited way.
Namely, whenever an edge $(u,v)$ is inserted and the insertion of the edge causes a distance decrease in the tree, we do \textit{not} update the currently maintained distance estimates.
Still the inserted edge may impact the distance estimates in later stages by preventing some estimates from increasing after further deletions. 

While it is easy to see that this change keeps the running time roughly the same as in the decremental setting, analyzing the correctness is a nontrivial challenge.
This is because the existing analyses of a monotone ES-tree work under specific structural assumptions and do not generalize to any construction. 
Specifically, while~\cite{henzinger2014} analyzed the stretch incurred by running monotone ES-trees on a hopset, the proof relied on the properties of the specific hopset used in their algorithm.
Since the hopset we use is quite different, we need a different analysis, which combines the static hopset analysis, with the ideas used in \cite{henzinger2014}, and also take into account the stretch incurred due to the fact that the restricted hopsets are maintained on the scaled graphs. Note also that our main hopset is not a simply a decremental maintenance of hopsets of \cite{elkin2017}, as our estimates are obtained from a \textit{sequence of hopsets} and insertions in one scale introduce insertions in the next scale. This is why we need a new argument and cannot simply rely on arguments in \cite{henzinger2014} and \cite{elkin2017}.

\paragraph{Putting it together.} We now go back to the setting of Lemma \ref{lem:single-scale}, and use a procedure similar to Algorithm \ref{alg:static_hopset}. Given a $2^j$-restricted hopset $\bar{H}_j=H_1 \cup... \cup H_{j}$ for distances up to $2^{j}$, we can now construct a graph $G^j$ by applying the scaling of Lemma \ref{lem:rounding} to $G \cup \bar{H}_j$ and setting $R=2^j$, $\ell=2\beta+1$. Then we can efficiently maintain an $\ell$-restricted hopset on $G^j$. Then by Lemma \ref{lem:single-scale} we can use this to update $H_{j+1}$. Importantly, $\ell$ is independent of $R$, and thus we can eliminate the factor $R$ to get $\tilde{O}(\beta mn^{\rho})$ total update time. Our final algorithm is a hierarchical construction that maintains the restricted hopsets on scaled graphs and the original graph simultaneously.

\paragraph{Stretch and hopbound analysis.} 
As discussed, applying the path-doubling idea to the hopset analysis is straightforward in static settings (and can be to some extent separated from the rest of the analysis) as is the case in \cite{elkin2019RNC}. 
However in our adapted decremental hopset algorithm this idea needs to be combined with the properties of the monotone ES tree idea and the fact that distance estimate are obtained from a sequence of hopsets on the scaled graph. In particular, in our stretch analysis we need to divide paths into smaller segments, such that the length of some segments is obtained from smaller iterations $i$, and the length of some segments are obtained from this combination of monotone ES tree estimates based on path doubling and scaling. We need a careful analysis to show that the stretch obtained from these different techniques combine nicely, which is based on a threefold inductive analysis:
\begin{enumerate}
    \item An induction on $i$, the iterations of the base hopset, which controls the sampling rate and the resulting size and hopbound tradeoffs.
    \item An induction on the scale $j$, which allows us to cover all ranges of distances $[2^j,2^{j+1}]$ by maintaining distances in the appropriate scaled graphs.
    \item An induction on time $t$ that allows us to handle insertions by using the estimates from previous updates in order to keep the distances monotone. 
\end{enumerate}
The overall stretch argument needs to deal with several error factors in \textit{addition to} the base hopset stretch. First, the error incurred by using hopsets for smaller scales, which we deal with by maintaining our hopsets by setting $\epsilon'=\frac{\epsilon}{\log n}$. This introduces polylogarithmic factors in the hopbound. The second type of error comes from the fact that the restricted hopsets are maintained for scaled graphs, which implies the clusters are only approximately maintained on the original graph. This can also be resolved by further adjusting $\epsilon'$. Finally, since we use an idea similar to the monotone ES tree of \cite{henzinger2014, henzinger2016}, we may set the level of nodes in each tree is to be larger than what it would be in a static hopset. But we argue that the specific types of insertions in our algorithm will still preserve the stretch. At a high-level this is because in case of a decrease we use an estimate from time $t-1$, which we can show inductively has the desired stretch. We note that while the monotone ES tree is widely used, we always need a different construction-specific analysis to prove the correctness.

\paragraph{Technical differences with previous decremental hopsets.} We note that while the use of monotone ES tree and the structure of the clusters in our construction are similar to \cite{henzinger2014}, our algorithm has several important technical differences. First, our hopset algorithm is based on different base hopset with a polylogarithmic hopbound, which as noted is crucial for obtaining near-optimal bounds in most of our applications. Additionally, we use a different approach to maintain the hopset efficiently by using path doubling and maintaining restricted hopsets on a sequence of \textit{scaled graphs}. Among other things, in \cite{henzinger2014} a notion of \textit{approximate ball} is used that is rather more lossy with respect to the hopbound/stretch tradeoffs. By maintaining restricted hopsets on scaled graphs, we are also effectively preserving approximate clusters/bunches with respect to the original graph, but as explained earlier, the error accumulation combines nicely with the path-doubling idea. Moreover, \cite{henzinger2014} uses an edge sampling idea to bound the update time, which we can avoid by utilizing the sampling probability adjustments in \cite{elkin2019RNC}, and the ideas in \cite{roditty2004}. Finally, our algorithm is based on maintaining the clusters up to a low hop, whereas they directly maintain bunches/balls.

\subsection{Applications in Decremental Shortest Paths}
Our algorithms for maintaining approximate distances under edge deletions are as follows. First, we maintain a $(\beta,1+\epsilon)$-hopset. Then, we use the hopset and Lemma \ref{lem:rounding} to reduce the problem to the problem of approximately maintaining short distances from a single source.
For our application in MSSP and APSP the best update time is obtained by setting the hopbound to be polylgarithmic whereas for SSSP the best choice is $\beta=2^{\tilde{O}(\sqrt{\log n})}$.
Using this idea for SSSP and MSSP mainly involves using the monotone ES tree ideas described earlier.
Maintaining the APSP distance oracle is slightly more involved but uses the same techniques as in our restricted hopset algorithm. This algorithm is based on maintaining Thorup-Zwick distance oracle \cite{TZ2005} more efficiently. At a high-level, we maintain \textit{both} a $(\beta, 1+\epsilon)$-hopset and Thorup-Zwick distance oracle simultaneously, and balance out the time required for these two algorithms. The hopset is used to improve the time required for maintaining the distance oracle from $O(mn)$ (as shown in \cite{roditty2004}) to $O(\beta mn^{1/k})$, but with a slightly weaker stretch of $(2k-1)(1+\epsilon)$. Querying distances is then the same as in the static algorithm of \cite{TZ2005}, and takes $O(k)$ time. 

\section{Decremental Hopset}

In this section we provide several decremental hopset algorithms with different tradeoffs. 
The starting-point of our constructions are the static hopsets described in Section \ref{sec:static_hopset}. But in order to get an efficient dynamic algorithm, we need to modify this construction in several ways. First, in Section \ref{sec:restricted_hopset}, we explain how we can adapt ideas by Roditty-Zwick \cite{roditty2004} to obtain an algorithm for computing a $d$-restricted hopset. The total running time of this algorithm is $O(dmn^{\rho})$ (where $\rho <1$ is a constant). While existentially this construction matches the state-of-the-art static hopsets with respect to size and hopbound tradeoffs, the update time is undesirable for large values of $d$. 

We will then provide another hopset algorithm with total running time of $\tilde{O}(mn^{\rho})$, by simultaneously maintaining a sequence of restricted hopsets using scaling and path-doubling ideas.
Recall that our algorithm maintains a sequence of graphs $H_0, \ldots, H_{\log (nW)}$, where for each $1 \leq j \leq \log (nW)$ , $H_0 \cup \ldots \cup H_j$ is a $2^j$-restricted hopset of $G$ ($W$ is the aspect ratio). Instead of computing each $H_j$ separately, we use $G \cup \bigcup_{r=0}^{j-1} H_{r}$ to construct $H_j$.
We observe that at the cost of some small approximation errors, any path of length $\in [2^{j-1}, 2^j)$ in $G$ can be approximated by a path of at most $2\beta+1$ hops in $G \cup \bigcup_{r=0}^{j-1} H_{r}$. 
To use this idea we will prove the following main lemma as a building block for our final hopset.
\singlescale*

There are two main challenges that we need to address for proving this lemma. First, we would like to make the running time independent of the scale bound $2^j$, which is what we would get by directly using the algorithm of \cite{roditty2004}.
To that end, we are going to run our algorithm on a scaled graph, which would allow us to only maintain distances up to depth $O(\beta / \epsilon)$.
This relies on having the $2^j$-restricted hopset $\bar{H}_j$, which allows us to maintain the hopset $\bar{H}_{j+1}$.
Second, while $G$ is undergoing deletions, $H_j$ may be undergoing edge \emph{insertions} incurred by restricted hopset edges added for smaller scales. In Section \ref{sec:new_hopset} we explain how such insertions can be handled using the monotone ES tree algorithm (proposed by \cite{henzinger2014}).
In Section \ref{sec:ss_stretch} we use the properties of this algorithm to prove Lemma \ref{lem:single-scale}.



\subsection{Maintaining a restricted hopset} \label{sec:restricted_hopset}


In this section our goal is to maintain a decremental \textit{restricted} hopset. 
We start by adapting the decremental algorithm by \cite{roditty2004} that maintains the Thorup-Zwick distance oracles \cite{TZ2005} with stretch $(2k-1)$ for pairs of nodes within distance $d$ in $\tilde{O}(dmn^{1/k})$ total time, but we use it to obtain a $d$-restricted hopset. In particular, using ideas in \cite{roditty2004}, and by restricting the shortest path trees up to depth $d$, we can maintain a variant of the hopset defined in Lemma \ref{lem:bunches} in which the hopset guarantee only holds for nodes within distance $d$.

In our adaptation of \cite{roditty2004}, we make the following two modifications: First, we change the sampling probabilities based on the hopset algorithm described in Section \ref{sec:static_hopset}.
Second, in addition to computing clusters we also add edges for each cluster that forms the hopset. This leads to a $(d, \beta, 1+\epsilon)$-hopset, but the update time is $\tilde{O}(dmn^{\rho})$.


We can then use the algorithm of Roditty and Zwick \cite{roditty2004} to maintain the clusters and bunches. At a high-level, the idea is to maintain Even-Shiloach \cite{ES} trees for each node $u \in A_i \setminus A_{i+1}$ to compute the cluster $C(u)$. The running time in the static algorithm can be bounded using the fact that each node overlaps with at most $\tilde{O}(n^{\rho})$ clusters (see the modified Dijsktra's algorithm description in Appendix \ref{app:static_hopset}). In dynamic settings this would translate to how many ES trees each node overlaps with. While these overlaps can be bounded at any point in time, it does not immediately hold for a sequence of updates, since nodes may keep leaving and joining clusters. 
Note that while the clusters we consider are slightly different than what was used in~\cite{roditty2004} (even if we ignore the difference in sampling probabilities), we use a \emph{subset} of the clusters defined in~\cite{roditty2004}.

We say that clusters are bounded by depth $d$ when $v \in C(u)$ if it satisfies Definition \ref{def:bunches}, and $d_G(v,u) \leq d$.
By extending the techniques \cite{roditty2004} we can get the following lemma for maintaining clusters (and hence bunches) up to a certain depth. At high-level this is done by computing the shortest path trees rooted at the cluster center up to depth $d$. Moreover, proof of this lemma relies on bounding the number of clusters overlapping each nodes over a sequence of updates. The details of this proof can be found in Appendix \ref{app:restricted_hopset}.

\begin{lemma}\label{lem:dec_bound_cluster}
There is an algorithm for maintaining clusters $C(w)$ (as in Definition \ref{def:bunches}) up to depth $d$ for all nodes $w \in A_i$, $0 \leq i \leq k+1/\rho+1$, such that for every $v \in V$ the expected total number of times the edges incident on $v$ are scanned over all trees (i.e.~trees on $C(w), w \in A_i$) is $O({d}/q_i)$, where $q_i$ is the sub-sampling probability.
 \end{lemma}


It is easy to see that the stretch analysis of \cite{elkin2019RNC} extends to a $d$-restricted hopset by only restricting the analysis to pairs of nodes within distance at most $d$. As we will see, in our main construction (Section \ref{sec:new_hopset}) we will need a completely new stretch analysis.

By combining the analysis of the modified Dijkstra algorithms of \cite{TZ2005}, Lemma \ref{lem:bunches}, and Lemma \ref{lem:dec_bound_cluster}, we can show that a  $d$-restricted hopset with the following guarantees can be constructed:

\begin{theorem} \label{thm:restricted_hopset}
Fix $\epsilon > 0, k \geq 2$ and $ \rho \leq 1$.
Given a graph $G=(V,E,w)$ with integer and polynomial weights, subject to edge deletions we can maintain a $(d, \beta, 1+\epsilon)$-hopset, with $\beta = O\left((\frac{1}{\epsilon} \cdot (k+1/\rho))^{k+1/\rho+1}\right)$ in $O(d (m+n^{1+\frac{1}{2^k-1}})n^{\rho})$ total time. 
The algorithm works correctly with high probability.
\end{theorem}

We review the algorithm of \cite{roditty2004}, and its analysis in Appendix \ref{app:restricted_hopset}. In Section \ref{sec:new_hopset} we will use a similar algorithm (that also handles insertions) for our new hopset construction, presented in Algorithm \ref{alg:restricted_hopset}, so we do not repeat the algorithm here.

\subsection{Restricted hopsets with improved running time}\label{sec:new_hopset}
As discussed, the algorithm described in Section \ref{sec:restricted_hopset} has a large update time for $d$-restricted hopsets, when $d$ is large.
In this section we provide a new hopset algorithm that is based on maintaining these restricted hopsets on a sequence of scaled graphs, and we show how this improves the update time, in exchange for a small (polylogarithmic) loss in the hopbound. In this construction, we rely on hopset edges added for smaller distance scales, and 
combined with a known technique-the monotone ES tree ideas-that are needed to handle insertions.
We note that while this general technique is used before, analyzing the correctness of monotone ES tree for each structure requires an analysis that depends on the specific construction.

\paragraph{Handling edge insertions}

\newcommand{\dest}{L}
First we explain the monotone ES tree idea and how it can be used for maintaining single-source shortest path up to a given depth $D$, and then combine these with the restricted hopset algorithm in Theorem \ref{thm:restricted_hopset}.

Using the monotone ES tree ideas may impact the stretch, and clearly do not apply to all types of insertions but only for insertion of certain structural properties.  In Section \ref{sec:ss_stretch}, we will prove that specifically for the insertions in our restricted hopset algorithm the stretch guarantee holds.

We show how to handle edge insertions by using a combining of the monotone ES-tree algorithm~\cite{henzinger2014} (and further used in the hopset construction of \cite{henzinger2016}). 
The idea in a monotone ES tree is that if an insertion of an edge $(u,v)$ causes the level of a node $v$, denoted by $L(v)$, in a certain tree to decrease, we will not decrease the level. In this case we say the edge $(u,v)$ and the node $v$ are \textit{stretched}. More formally, a node $v$ is stretched when  $L(v) > \min_{(x,v) \in E} \dest(x) + w(x, v)$. The complete algorithm is presented in Algorithm~\ref{alg:estree} in Appendix \ref{app:monotone_es}. The update time analysis is straightforward, and summarized in the following lemma:

\begin{lemma} \label{lem:es_time}
Give a graph undergoing edge-deletions, and a set of edge insertions, a monotone ES tree can be maintained in $O( (m+\Delta) D)$ overall update time on a graph with $m$ edges, where $\Delta$ is the number of edge insertions.
\end{lemma}
\begin{proof}[Proof sketch.]
The running time analysis of the algorithm follows based on an argument similar to the analysis of the classic ES tree algorithm \cite{ES, king1999}.  The total time for updating distances up to a depth $D$ is $O( (m+\Delta) D)$, roughly speaking, since the edges incident to each node $v$ are scanned any time level of each node changes. Since distances can only increase there can be at most $D$ times for nodes with depth at most $D$ to the source. Furthermore, $\Delta$ is the number of added edges that also need to be scanned in each update. By summing over all edges incident to all nodes the claim follows.
\end{proof}

\paragraph{Restricted hopsets with insertions.}

Next, given a sequence of deletes $E^-$, and insertions $E^+$
in Algorithm \ref{alg:restricted_hopset} we describe how the algorithm of \cite{roditty2004} is modified to handle these insertions by combining it with the monotone ES tree algorithm of Lemma \ref{lem:es_time} (Algorithm \ref{alg:estree} in Appendix \ref{app:monotone_es}) for each tree inside a cluster. Using this we can bound the update time, however proving the stretch is more involved, and depends on the specific structure of insertions, and it does not hold for any set of insertions.  

Our goal is to maintain the hopset algorithm of Section \ref{sec:static_hopset}. Recall that we start by sampling sets $V=A_0 \supseteq A_1 \supseteq ... \supseteq A_{k+1/\rho+1}=\emptyset$ initially, with sampling probabilities $q_i=\max(n^{-2^i \cdot \nu}, n^{-\rho})$, where $0 < \rho \leq 1/2$ is a parameter.  The sets remain unchanged during the updates.
Next, we need to maintain values $d(v, A_i), 1 \leq i \leq k+1/\rho+1$ for all nodes $v \in V$, and use these values for maintaining $p(v)$. For this we can simply maintain a monotone ES tree (using Algorithm \ref{alg:estree}) rooted at a dummy node $s_i$ connected to all nodes in $A_i$ up to depth $d$, in total time $O(dm)$. We denote the estimate obtained by maintaining this distance by $L(v,A_{i+1})$. The pivots $p(v), \forall v \in V$ can also be maintained in this process.

Next we need to maintain the clusters. Handling nodes that leave a cluster is simpler. Recall that in the static construction for $z \in A_i\setminus A_{i+1}$ we have $v \in C(z)$ if and only if $d(z,v) < d(v,A_{i+1})$, but here we maintain \textit{approximate} bunches and clusters based on monotone ES tree estimates. After each deletion, for each node $v$ and the cluster centers $z$ we first check whether the distance estimate $L(z,v)$ has increased.  If $L(z,v) \geq \frac{L(v, A_{i+1})}{1+\epsilon}$, $v$ will be removed from $C(z)$. There is a slight technicality here for ensuring that the size of the bunches are still bounded. Instead of directly using $L(v,A_{i+1})$, we maintain approximate bunches with radius $\frac{L(v,A_{i+1})}{1+\epsilon}$, which ensures that the size of the bunches are bounded as these are subsets of the original bunches on $G$ (as argued in Lemma \ref{lem:dec_bound_cluster}). This $\epsilon$ parameter is rescaled later appropriately for obtaining the final estimate.   
The more subtle part is adding nodes to new clusters. For each $0 \leq i <k+1/\rho+1$, we define a set $X_{i}$ consisted of all vertices whose distance to $A_i$ is increased as a result of a deletion, but where this distance is still at most ${d}$. The sets $X_{i}$ can be computed while maintaining $L(v,A_i)$. This can also be done by maintaining a single tree rooted at a dummy node $s_i$.

 Note that a node $v$ would join $C(w)$ only after an increase in $L(v,A_{i+1})$. Using this observation, we can use the modified Dijkstra algorithm (also used in static hopset construction in Appendix \ref{app:static_hopset}) which can be summarized as follows: when we explore neighbors of a node $x$, we only relax an edge $(x,y)$ if $L(x,y)+w(x,y) < \frac{L(x,A_{i+1})}{1+\epsilon}$. 
 Hence in each iteration $i$, after each deletion for every $v \in X_{i+1}, z \in B_i(u) \setminus B_i(v)$, and each edge $(u,v) \in E$ we check if $L(z,u) +w(u,v) < \frac{L(v, A_{i+1}}{1+\epsilon})$. If yes, then $v$ joins $C(z)$, and $v$ is pushed to a priority queue $Q(z)$. This priority $Q(z)$ stores the distances in each tree $T(z)$ rooted at $z$.
 These nodes join clusters $T(z)$, but there may be other nodes that also need to join $C(z)$ as a result of this change.
Hence after this initial phase, for each $z \in A_i \setminus A_{i+1}$ where $Q(z) \neq \emptyset$, we run the modified Dijkstra's algorithm. 
A summary of this algorithm is presented in Algorithm \ref{alg:restricted_hopset}. Note that the input to this algorithm is a graph $G$, distance bound $d$, a set of edges $E^-$ deleted (or updated by distance increase), and a set of insertions $E^+$. 

By combining these two algorithm we can keep the running time the same as in Theorem \ref{thm:restricted_hopset} despite the insertions:

\begin{theorem}\label{thm:monotone_es_time}
Assume that we are given a set of $\Delta$ updates including a set $E^-$ of deletions and a set $E^+$ of insertions, and parameters $d$, and $\epsilon$, w.h.p.~the total update time of Algorithm \ref{alg:restricted_hopset} is $O((m+\Delta+n^{1+\frac{1}{2^k-1}})dn^{\rho})$. 
\end{theorem}
\begin{proof}[Proof sketch.]
The proof of this theorem is almost exactly the same as the proof of Theorem \ref{thm:restricted_hopset}. We rely on Lemma \ref{lem:dec_bound_cluster} again to show the over a sequence of updates, for each iteration $i$, each node $v$ is only in $\tilde{O}(1/q_i)$ clusters. Since the monotone ES tree ensures that the insertions do not reduce the level of a node, the number of times edges incident to $v$ are scanned is still $\tilde{O}(d/q_i)= \tilde{O}(dn^{\rho})$. We now have $m+\Delta$ edges, and the theorem follows by summing overall nodes.
\end{proof}
We showed that we can handle a set of insertions within the same the running time. But as discussed, directly using algorithm \ref{alg:restricted_hopset} still does not lead to our desired update time. 
Therefore in the rest of this section we described how we can get improved running time by using this algorithm to maintain a hierarchical construction of restricted hopsets on a sequence of scaled graphs. As explained one idea is that we can add hopset edges for smaller scales and use the added edges in computing distances for larger scales. Once we specify the set of insertion into each of the \textit{scaled graphs} considered, we will show that such insertions will also preserve the hopset stretch (with small polylogarithmic overhead) in the \textit{original graph}.\\

\begin{algorithm}[H]\small
\caption{Monotone $d$-restricted hopset. Adaptation of \cite{roditty2004}. }
\label{alg:restricted_hopset}
\SetKwProg{Fn}{Function}{}{}
Sample sets $V=A_0 \supseteq A_1 \supseteq ... \supseteq A_{k+1/\rho+1}=\emptyset$.\\ 
\Fn{\textsc{UpdateClusters}$(G,E^-,E^+,d)$}{
 Add edges $(x,y) \in E^+$ to any tree $T(z)$ s.t.~$(x,y) \in T(z)$\\
\For{$i=0$ to $k+1/\rho+1$}{
  $\mathcal{C} = \emptyset$.\\
  Remove edges $E^-$ from the ES tree maintaining distances $L(\cdot, A_{i+1})$\\
  Remove hopset edges $(z,v)$, and remove $v$ from $T(z)$ where $L(z,v) \geq \frac{L(v,A_{i+1})}{1+\epsilon}$\\
  $X_{i+1} := $ set of nodes whose distances to $A_{i+1}$ have increased due to removal of $E^{-}$, yet remained at most $d$\\
  \For{$\forall v \in X_{i+1}$}{
    \For{$(u,v) \in E$}{
        \For{$\forall z \in B_i(u)\setminus B_i(v)$}
            {
             \If{$L(z,u) +w(u,v) < \frac{L(v,A_{i+1})}{1+\epsilon}$}{
                $\mathcal{C}= \mathcal{C} \cup \{ z\}$\\
                \textsc{Relax}($(Q(z), u,v)$)\tcc{Update the estimate from $z$ to $v$}
                }
            }
        }
    }
    \For{$\forall z \in \mathcal{C}$}
    {
        \textsc{Dijkstra}$(z)$
    }
    }   
    return $(E^-,E^+)$
}

\Fn{\textsc{Dijkstra}$(z)$}{
    \While{$Q(z) \neq \emptyset$}{
    $u= \textsc{ExtractMin}(Q(z))$\\
    $B(u)= B(u) \cup \{z\}$\\
        \For{$\forall(u,v) \in E: z \not \in B(v)$}{
            \If{$L(z,u)+w(u,v) < \frac{L(v,A_{i+1})}{1+\epsilon}$}{ 
                \textsc{Relax}$(Q(z), u,v)$\tcc{Update the estimate from $z$ to $v$} 
            }
        }
    }
}
\Fn{\textsc{Relax}$(Q(z),u,v)$}{
    \tcc{Distances $L(z,v)$ for each tree $T(z)$ are maintained in $Q(z)$}
    $d' := L(z,v)+ w(z,v)$\\
    \If{$d' \leq d$}{
        \If{$v \in Q(z)$}{
            \textsc{decrease-key}$(Q(z), v, d')$}
    \ElseIf{$L(z,u) > d'$}{
            \textsc{Insert}$(Q(z), v, d')$
        }
    Add node $v$ to $T(z)$\\    
    \textsc{InsertEdge}($T(z),(z,v), d'$)\tcc{As defined in Algorithm \ref{alg:estree} in Appendix \ref{app:monotone_es}}
    $E^+= E^+ \cup \{ (z,v)\}$\\
    Add $(z,v)$ to $E^-$ if $L(z,v)$ has increased.
    }
}
\end{algorithm}

\paragraph{Path doubling and scaling.}
We first state the path doubling idea more formally for a \textit{static hopset} in the following lemma. However for utilizing this idea dynamically we need to combine it with other structural properties of our hopsets.


\begin{lemma}\label{lem:hop_doubling}
Given a graph $G=(V,E)$, $0 < \epsilon_1 <1$, the set of  $(\beta, 1+\epsilon_1)$-hopsets $H_r, 0 \leq r < j$ for each distance scale $(2^r, 2^{r+1}]$, provides a $(1+\epsilon_1)$-approximate distance for any pair $x, y \in V$, where $d(x,y) \leq 2^{j+1}$ using paths with at most $2\beta+1$ hops.
\end{lemma}
\begin{proof} 
We can show this by an induction on $j$. Let $\pi$ be the shortest path between $x$ and $y$ on . Then $\pi$ can be divided into two segments, where for each segment there is a $(1+\epsilon_1)$-stretch path using edges in $G \cup \bigcup_{r=0}^{j-1} H_r$. Let $[x,z]$ and $[z',y]$ be the segments on $\pi$ each of which has length at most $2^{j-1}$. In other words, $z$ is the furthest point from $x$ on $\pi$ that has distance at most $2^{j-1}$, and $z'$ is the next point on $\pi$.  Then we have,
\begin{align*}
 d^{(2\beta+1)}_{G \cup \bigcup_{r=1}^{j-1}H_r} (x,y) &\leq  [d^{(\beta)}_{G \cup \bigcup_{r=1}^{j-1}H_r}(x,z)+ w(z,z')+ d^{(\beta)}_{G \cup \bigcup_{r=0}^{j-1}H_r}(z',y) ]\\
& \leq (1+\epsilon_1) d_{G}(x,z) + w(z,z')  + (1+\epsilon_1) d_{G}(z',y) \\
&\leq  (1+\epsilon_1) d_{G} (x,y)
\end{align*}
\end{proof}

This implies that it is enough to compute $(2\beta+1)$-hop limited distances in restricted hopsets for \textit{each} scale. For using this idea in dynamic settings we have to deal with some technicalities. We should show that we can combine the rounding with the modification needed for handling insertions.

We define a scaled graph using Lemma~\ref{lem:rounding}
as follows: $G^j := \textsc{Scale}(G \cup \bigcup_{r=0}^{j} H_{r}, 2^j, \epsilon_2, 2\beta +1)$. Here we set $R=2^j, \ell=2\beta+1$, and $\epsilon_2$ is a parameter that we tune later. We first describe the operations performed on this scaled graph. We then explain how we can put things together for all scales to get the desired guarantees. The key insight for scaling $G \cup \bigcup_{r=0}^{j} H_{r}, 2^j$ is that we can obtain $H_{j+1}$ by computing an $O(\ell)$-restricted hopset of $G^j$ (using the algorithm of Lemma~\ref{lem:single-scale}) and scaling back the weights of the hopset edges.

In addition to the graph $G$ undergoing deletions, our decremental algorithm maintains the following data structures for each $1 \leq j \leq \log (nW)$:
\begin{itemize}
    \item The set $\bar{H}_j= \bigcup_{r=0}^j H_r$, union of all hopset edges for distance scales up to $[2^j,2^{j+1}]$.
    \item The scaled graphs $G^1, ..., G^j$.
    \item Data structure obtained by constructing an $O(\beta/\epsilon_2)$-restricted hopset on ${G}^j$ by running Algorithm \ref{alg:restricted_hopset} for the appropriate parameter $\epsilon_2<1$. We denote this data structure by $D_j$.
\end{itemize}
The data structure $D_j$ is maintained by running  Algorithm \ref{alg:restricted_hopset} on ${G}^j$, and maintaining the clusters and hence the bunches $B(v)$ for all $v \in V$. Given $D_j$, we can maintain $H_{j+1}$, where the edge weights in clusters are assigned by computing approximate distances based on the monotone ES tree on each cluster as follows: In a tree rooted at a cluster center $z$, we set the weight $w_j$ on an edge $(z,v)$ to be $\min^{j-1}_{r=1} \eta( 2^r,  \epsilon_2) L_r(z,v)$, where $L_r(z,v)$ is the level of $v$ on $G^r$ after running the monotone ES tree up to depth $D=\lceil \frac{2(2\beta+1)}{\epsilon_2} \rceil$. We the maintain a restricted hopset on the scaled graph $G^j$, and by \textit{unscaling} its weights we get $H_{j+1}$.
Note that we never underestimate any distances. The rounding in Lemma \ref{lem:rounding} does not underestimate the distances, and if the edge is stretched that means we are assigning a weight larger than what is obtained by the rounding.

Once each data structure $D_j$ is initialized with a graph, it can execute a single operation $\textsc{Update}(E^-, E^+)$, which updates the maintained graph by removing the edges of $E^-$ and adding edges $E^+$ by running Algorithm \ref{alg:restricted_hopset}. The set $E^-$ is the set the edges corresponding to nodes leaving clusters. 
The operation returns a pair of edges $(E^-, E^+)$ that are edges that should be removed or added from $D_j$. Additionally, by multiplying these distance by $\eta(2^j, \epsilon_2)$ for the appropriate $\epsilon_2$, we can recover a pair $(H^-, H^+)$ of edge sets, where $H^-$ is the set of edges that are removed from the hopset and $H^+$ is the set of edges added to the hopset as a result of the update.
Note that a change in the weight of a hopset edge is equivalent to removing the edge and adding it with a new weight.

In Algorithm \ref{alg:main} we update the data structures described as follows: we run Algorithm \ref{alg:restricted_hopset} for distances bounded by $d= \lceil\frac{ 2(2 \beta+1)}{\epsilon_2}\rceil$ starting on $j= 0,..., \log W$ in increasing order of $j$ to compute hopset edges $H_j$. After processing all the changes in scaled graph $G^j$, we add the inserted edges to $G^{j+1}$. Then we process the changes in $G^{j+1}$ by running the algorithm of Section \ref{sec:restricted_hopset} and repeat until all distance scales of covered. As explained, when the distances increase a node may join a new cluster which will lead to a set of insertions in $H$ and in turn insertions in a sequence of graphs $G^j$. 
We use an argument similar to Lemma \ref{lem:dec_bound_cluster} on each scaled graph to get the overall update time. In a way we can see the added edges passed to each scale as a set of batch distance increases, between the corresponding endpoints. This means we are not exactly in the setting of \cite{roditty2004} where only one deletion occurs at each time, but the exact same analysis as in Lemma \ref{lem:dec_bound_cluster} still holds. 

\begin{algorithm}[h]
\caption{Updating the hopset after deleting an edge $e$.}
\label{alg:main}
\SetKwProg{Fn}{Function}{}{}
Input: $0<\epsilon, 0< \epsilon_2 <1$, set $d= \lceil\frac{ 2(2 \beta+1)}{\epsilon_2}\rceil$.\\
$(E^-, E^+) := (\{e\}, \emptyset)$\\
\For{$j = 0, \ldots, \lfloor \log W \rfloor$}{
     $(E^-, E^+) :=  \textsc{UpdateClusters}(G^j, E^{-}, d, \epsilon)$\tcc{Run Algorithm \ref{alg:restricted_hopset} on ${G}^j$}
    Update $H_{j+1}$ by unscaling weights of $E^+$  and removing $E^-$ (Lemma \ref{lem:rounding}) \tcc{add edges for  the next scale}
     Update $G^{j+1}$ based on Lemma \ref{lem:rounding} to reflect changes to $H_{j+1}$
  }
\end{algorithm}

We summarized the algorithm obtained by maintaining this data structure over all scales in Algorithm \ref{alg:main}. Note that we need to update both the restricted hopsets $D_j$ on the scaled graphs and the hopset $H_j$ obtained by scaling back the distances using Lemma \ref{lem:rounding}.

\paragraph{Running time (proof of Lemma \ref{lem:single-scale}).} 
We can now put all the steps discussed to maintain the data structure of Lemma \ref{lem:single-scale}. 
In particular, for obtaining a $2^{j+1}$-restricted hopset, we maintain the data structure of Lemma \ref{lem:single-scale} on ${G}^j$ for each cluster rooted at a node $z \in A_i \setminus A_{i+1}$ and by setting $\ell=2\beta+1$. By using Lemma \ref{lem:es_time} and Theorem \ref{thm:restricted_hopset} to compute $d$-restricted hopsets for $d=O(\beta /\epsilon)$. When weights are polynomial we get the running time of $\tilde{O}(\frac{\beta}{\epsilon}(m+\Delta)n^\rho)$, where $\Delta$ is the overall number of hopset edges added over all updates.

\subsection{Hopset stretch} \label{sec:ss_stretch}
In this section, we first prove the stretch incurred for a single-scale by combining properties of the monotone ES-tree algorithm (incorported to Algorithm \ref{alg:restricted_hopset}) with the static hopset argument and the rounding framework. We will then show that by setting the appropriate parameters we can prove the overall stretch and hopbound tradeoffs described in Lemma \ref{lem:single-scale}.

In the following, we extend the static hopset argument to dynamic settings. We use the path doubling observation in Lemma \ref{lem:hop_doubling} and properties of monotone ES tree described to prove the stretch incurred in each scale. We denote the stretch of $\bar{H}_j$ to be $(1+\epsilon_j)$.  Then for getting the final stretch and hopbound we will set the parameters $\epsilon_2=\epsilon'$ (error incurred by rounding), and $\delta= \frac{\epsilon}{8(k+1/\rho +1)}$. 

The stretch argument is based on a threefold induction on $i$, $j$-th scale, and time $t$. By fixing $i,j,t$, and a source node $s$, we show that there is a $(1+\epsilon_j)$-stretch path between $s$ and any other node with $\beta$ hops (or if we are using previous scale $2 \beta+1$-hops) such that based each segment of this path has the desired stretch based on the inductive claim on one of these three factors. At a high level induction on $i$ and $j$ follows from static properties of our hopset. To show that bounded depth monotone ES tree maintains the approximate distances, we note that any segment of the path undergoing an insertion consistent of a single shortcut and the weight on such an edge is a distance estimate between its endpoints. 

It is easy to see that we never underestimate distances. Roughly speaking, we either obtain an estimate from rounding estimates obtained from smaller scales, which is an upper bound on the original estimate, or we ignore a distance decrease.

\begin{theorem} \label{thm:single_stretch}
Given a graph $G=(V,E)$, assume that we have maintained a $(2^j, \beta, (1+\epsilon_j))$-restricted hopset $\bar{H}_j$,  and let $H_{j+1}$ be the hopset obtained by running Algorithm \ref{alg:main}  for any given $0 \leq \epsilon_2 < 1$ on $G \cup \bar{H}_{j}$. 
Fix $0 < \delta \leq \frac{1}{8(k+1/\rho+1)}$, and consider a pair $x,y \in V$ where $d_{t,G}(x,y) \in [2^j, 2^{j+1}]$. Then for $0 \leq i \leq k+1/\rho+1$, either of the following conditions holds:
\begin{enumerate}
    \item $d^{((3/\delta)^i)}_{G \cup \bar{H}_{j+1}}(x,y) \leq (1+8\delta i)(1+\epsilon_j)(1+\epsilon_2) d_{t,G}(x,y)$ or,
    \item There exists $z \in A_{i+1}$ such that, 
\[d^{((3/\delta)^i)}_{G \cup \bar{H}_{j+1}} (x,z) \leq 2(1+\epsilon_j)(1+\epsilon_2) d_{t,G}(x,y).\]
    \end{enumerate}
    Moreover, by maintaining a monotone ES tree on $G^{j+1}$ up to depth $\lceil \frac{2(2\beta+1)}{\epsilon_2} \rceil$, and applying the rounding in Lemma \ref{lem:rounding}, we can maintain ${(1+\epsilon_{j+1})}$-approximate single-source distances up to distance $2^{j+2}$ from a fixed source $s$ on $G$, where $1+\epsilon_{j+1} = (1+\epsilon_j)(1+\epsilon_2)^2(1+\epsilon)$ and $\beta=(3/\delta)^{k+1/\rho+1}$.
\end{theorem} 

\begin{proof}
We use a double induction on $i$ and time $t$, and also rely on distance computed up to scaled graph $G^j$. First, using these distance estimates for smaller scales, we argue that when we add an edge to $\bar{H}_{j+1}$ it has the desired stretch. 
Let $L_{t,j}(u,v)$ denote the level of node $v$ in the tree rooted at $u$ after running Algorithm \ref{alg:restricted_hopset} up to depth $D= \lceil \frac{2(2\beta+1)}{\epsilon}\rceil$ on graph $G^j$.
This proof is based on a cyclic argument: assuming we have correctly maintained distances up to a given scale using our hopset, we show how we can compute the distances for the next scale.
In particular, we first assume that based on the theorem conditions we are given $\bar{H}_j$ and have maintained all the clusters and the corresponding distances in $G^1,...,G^j$ with stretch $1+\epsilon_j$. This lets us analyze $H_{j+1}$. Then to complete the argument, we show how given the hopsets of scale $[2^j,2^{j+1}]$, we can compute approximate SSSP distances for the next scale based on the monotone ES tree on $G^{j+1}$.

First, in the following claim, we observe that the edge weights inserted in the latest scale have the desired stretch by using our assumption that all the shortest path trees on each cluster on $G^1,...,G^j$ are approximately maintained. We use such distance to add edges in each cluster to construct $H_{j+1}$, and we observe the following about the weights on these edges:

\begin{observation}\label{obs:monotone}
Let $v \in B(u)$ such that $d_{t,G}(u,v) \leq 2^{j+1}$. Consider an edge $(u,v)$ added to $H_{j+1}$ after running Algorithm \ref{alg:restricted_hopset} on $G^1,...,G^j$ for $D=\lceil \frac{2(2\beta+1)}{\epsilon_2} \rceil$ rooted at node $v$. Let $w_{j+1}(u,v):=\min^{j}_{r=1} \eta( 2^r,  \epsilon_2) L_r(u,v)$,  that is the unscaled edge weight. Then we have\ $d_{t,G}(u,v) \leq w_{j+1}(u,v) \leq (1+\epsilon_j)(1+\epsilon_2)d_{t,G}(u,v)$.
\end{observation}

This claim implies that the weights of hopset edges assigned by the algorithm correspond to approximate distance of their endpoints.
Let $d_{t,j}(x,y):= \min_{r=1}^j \eta(2^r, \epsilon_2)L_{t,j}(x,y)$ which would be the estimate we obtain by for distance between $x$ and $y$ after scaling back distances on $G^r, 1 \leq r \leq j$.
In other words this is the hop-bounded distance after running monotone ES tree on $G^j$ and scaling up the weights.

For any time $t$ and the base case of $i=0$, we have three cases. If $y \in B(x)$ then edge $(x,y)$ is in the hopset $H_{j+1}$, and by Observation \ref{obs:monotone} the weight assigned to this edge is at most $(1+\epsilon_j)(1+\epsilon_2)d_{t,G}(x,y)$. In this case the first condition of the theorem holds. Otherwise if $x \in A_1$, then $z=x$ trivially satisfies the second condition. Otherwise we have $x \in A_0/A_1$, and by setting $z=p(x)$ we know that there is an edge $(x,z) \in \bar{H}_j$ such that $d_{t,j}(x,z) \leq  (1+\epsilon_2)d_{G \cup \bar{H}_j}(x,y)$ (by definition of $p(x)$ and using the same argument as above). Hence the second condition holds. 

By inductive hypothesis assume the claim holds for $i$. Consider the shortest path $\pi(x,y)$ between $x$ and $y$. We divide this path into $1/\delta$ segments of length at most $\delta d_{t,G}(x,y)$ and denote the $a$-th segment by $[u_a,v_a]$, where $u_a$ is the node closest to $x$ (first node of distance at least $a \delta d_{t,G}(x,y)$) and $v_a$ is the node furthest to $x$ on this segment (of distance at most $(a+1)\delta d_{t,G}(x,y)$).

 We then use the induction hypothesis on each segment. First consider the case where for all the segments the first condition holds for $i$, then there is a path of $(3/\delta)^{i}(1/\delta) \leq (3/\delta)^{i+1}$ hops consisted of the hopbounded path on each segment. We can show that this path satisfies the first condition for $i+1$. In other words, 
 \[ d^{((3/\delta)^{i+1})}_{t,G \cup \bar{H}_{j+1}} (x,y) \leq \sum_{a=1}^{1/\delta}d^{((3/\delta)^{i})}_{t,G \cup \bar{H}_{j+1}} (u_a,v_a) +d^{(1)}_{t,G}(v_a,u_{a+1}) \leq (1+8\delta i)(1+\epsilon_j)(1+\epsilon_2)d_{t,G}(x,y)\]

Next, assume that there are at least two segments for which the first condition does not hold for $i$. Otherwise, if there is only one such segment a similar but simpler argument can be used. Let $[u_l, v_l]$ be the first such segment (i.e.~the segment closest to $x$, where $u_l$ is the first and $v_l$ is the last node on the segment), and let $[u_r, v_r]$ be the last such segment.


First by inductive hypothesis and since we are in the case that the second condition holds for segments $[u_l,z_l]$ and $[u_r,v_r]$, we have,
\begin{itemize}
\item  $d_{t, G \cup \bar{H}_{j+1}}^{((3/\delta)^i)} (u_l, z_l)  \leq 2(1+\epsilon_2)(1+\epsilon_j) d_{t,G}(u_{l},v_l)$, \textit{and,}
\item $d_{t,G \cup \bar{H}_{j+1}}^{((3/\delta)^i)} (v_r, z_r)  \leq 2(1+\epsilon_2)(1+\epsilon_j)d_{t,G}(u_r, v_r)$
\end{itemize}

Again, we consider two cases. First, in case $z_r \in B(z_l)$ (or $z_l \in C(z_r)$), we have added a single hopset edge $(z_r, z_l) \in \bar{H}_{j+1}$. Note that $d_{t,G}(z_r,z_l) \leq 2^{j+1}$, since $d_{t,G}(z_r,z_l) \leq d_{t,G}(x,y) \leq 2^{j+1}$. Hence by Observation \ref{obs:monotone} the weight we assign to $(z_r, z_l)$ is at most $(1+\epsilon_2)(1+\epsilon_j)d_{t,G}(z_r, z_l)$. 

On the other hand, by triangle inequality, and the above inequalities (which are based on the induction hypothesis) we get,
\begin{align} 
d^{(1)}_{\bar{H}_{j+1}}(z_l,z_r) &\leq (1+\epsilon_2)(1+\epsilon_j) d_G(z_l,z_r)\\
&\leq (1+\epsilon_2)(1+\epsilon_j) [d_{G \cup \bar{H}_{j+1}}^{((3/\delta)^i)}(u_l,z_l)+d_G(u_l,v_r)+d^{((3/\delta)^i)}_{G \cup \bar{H}_{j+1}} (z_r,v_r)]\label{eq:expand}
\end{align}


By applying the inductive hypothesis on segments before $[u_l, v_l]$, and after $[u_r,v_r]$, we have a path with at most $(3/\delta)^i$ for each of these segments, satisfying the first condition for the endpoints of the segment. Also, we have a $2(3/\delta)^i +1$-hop path going through $u_{l}, z_{l}, z_r, v_r$ that satisfies the first condition for $u_{l}, v_r$.

Putting all of these together, we argue that there is a path of hopbound $(3/\delta)^{i+1}$ satisfying the first condition.  In particular, we have (the subscript $t$ is dropped in the following),
\begin{align}
    d^{(3/ \delta)^{(i+1})}_{G \cup \bar{H}_{j+1}} (x,y) &\leq  \sum^{l-1}_{a=1} [d_{G \cup \bar{H}_{j+1}}^{((3/\delta)^i)} (u_a, v_a) +d^{(1)}_G(v_a,u_{a+1})] +d^{((3/\delta)^i)}_{G \cup \bar{H}_{j+1}} (u_l, z_l) \label{line:not_stretched}\\
    &+d^{(1)}_{\bar{H}_{j+1}} (z_l, z_r) + d^{((3/\delta)^i)}_{G \cup \bar{H}_{j+1}} (z_r, v_r)  +d^{(1)}_{G} (v_r, u_{r+1})\\
    &+\sum^{ 1/\delta }_{a=r+1} [d_{G \cup \bar{H}_{j+1}}^{((3/\delta)^i)} (u_a, v_a) +d^{(1)}_{G}(v_a,u_{a+1})] \label{line:triangle}\\
    &\leq (1+8\delta i)(1+\epsilon_j)(1+\epsilon_2) [d_{G}(x,u_l)+d_G(v_r,y)] + d_G(u_l,v_r)\\ 
    &+ (1+\epsilon_2)(1+\epsilon_j)[2d_{G}(u_l,z_l) + 2d_{G}(z_r,v_r)]\\
    &\leq (1+\epsilon_2)(1+\epsilon_j) [8 \delta d_{G}(x,y) + (1+ 8 \delta i) d_{G}(x,y)] \label{line:subpath}\\
    & \leq (1+8\delta (i+1))(1+\epsilon_2)(1+\epsilon_j) d_G(x,y)
\end{align}
In the first inequality we used the induction on $i$ for each segment, and triangle inequality. 

In the second inequality we are using the fact that nodes $u_j,v_j$ for all $j$ are on the shortest path between $x$ and $y$ in $G$, and we are replacing $d^{(1)}_{\bar{H}_{j+1}} (z_l, z_r)$ with inequality \ref{eq:expand}.
In line \ref{line:subpath} we used the fact that the length of each segment is at most $\delta \cdot d_G(x,y)$. Hence we have shown that the first condition in the lemma statement holds.

Finally, consider the case where $z_r \not \in B(z_l)$. If $z_l \not \in A_{i+2}$, we consider $z =p(z_l)$, where $z_l \in A_{i+2}$. We now claim that this choice of $z$ satisfies the second lemma condition. 

We have added the edge $(z_l, z)$ to the hopset. Since $z_r \not \in B(z_l)$, we have $d_{t-1,G}(z_l,p(z_l))\leq d_{t-1,G}(z_l,z_r) \leq d_{t,G}(x,y) \leq 2^{j+1}$. Therefore we can use Observation \ref{obs:monotone} on the edge $(z_l,p(z_l))$.
\begin{align}
    d^{(3/ \delta)^{(i+1})}_{G \cup \bar{H}_{j+1}} (x,y) &\leq  \sum^{l-1}_{a=1} [d_{G \cup \bar{H}_{j+1}}^{((3/\delta)^i)} (u_a, v_a) +d^{(1)}_G(v_a,u_{a+1})] +d^{((3/\delta)^i)}_{G \cup \bar{H}_{j+1}} (u_l, z_l)+(1+\epsilon_2)(1+\epsilon_j)d^{(1)}_{\bar{H}_{j+1}} (z_l, z)\\
    &\leq (1+8\delta i)(1+\epsilon_2)(1+\epsilon_j) d_{G}(x,u_l) +d^{((3/\delta)^i)}_{G \cup \bar{H}_{j+1}} (u_l, z_l)+(1+\epsilon_2)(1+\epsilon_j)d_{\bar{H}_{j+1}}(z_l, z_r)\\
     &\leq (1+8\delta i)(1+\epsilon_2)(1+\epsilon_j) d_{G}(x,u_l) +d^{((3/\delta)^i)}_{G \cup \bar{H}_{j+1}} (u_l, z_l)\\
     &+(1+\epsilon_2)(1+\epsilon_j)[2d^{((3/\delta)^i)}_{G \cup \bar{H}_{j+1}}(z_l, u_l)+d_G(u_l,v_r)+d^{(3/\delta)^i}_{G\cup \bar{H}_{j+1}}(v_r,z_r)]\\
    &\leq (1+ 8\delta i)(1+\epsilon_2)(1+\epsilon_j) d^{((3/\delta)^i)}_{G\cup \bar{H}_{j+1}}(x,v_r)+ 6\delta (1+\epsilon_j) d_{G}(x,y)\\
    & \leq 2(1+\epsilon_2)(1+\epsilon_j)d_G(x,y)
    \end{align}

In the last inequality we used the fact that we set $\delta <\frac{1}{8(k+1/\rho+1)}$ and thus $8\delta i <1$. The only remaining case is when $z_{\ell} \in A_{i+2}$, in which case a similar reasoning follows by setting $z=z_{l}$.

Finally, we prove that after adding hopset edges $H_{j+1}$ we can maintain approximate single-source shortest path distances from a given source $s$. This enables us to show that Observation \ref{obs:monotone} can be used for the next scale, i.e. that we can set the weights for the next scale by maintaining the clusters and $(1+\epsilon_{j+1})$ approximate distance rooted at a source $s$ when we have $d(s,v) \in [2^{j+1}, 2^{j+2}]$, and hence close the inductive cycle in the argument.

We run the monotone ES tree algorithm (Algorithm \ref{alg:estree} in Appendix \ref{app:monotone_es}) up to depth $\lceil \frac{2(2\beta+1)}{\epsilon} \rceil$ on all of the scaled graphs $G^1, ..., G^{j+1}$. We let set the distance estimate $d_{t,j+1}(s,v)$ to be $\min_{r} \eta( 2^r,  \epsilon_2) L_{t,r}(s,v)$ where $L_{t,r}(s,v)$ is the level of $v$ on $G^r$ on the ES tree up to depth $\lceil \frac{2(2\beta+1)}{\epsilon_2} \rceil$ rooted at $s$. Note that by running Algorithm \ref{alg:restricted_hopset} we are also maintaining the same distances on each cluster while also maintaining the nodes that leave and join a cluster. 
We analyze the estimate for any $v \in V$ such that $d_G(s,v) \leq 2^{j+2}$. W.l.o.g. assume that $d(s,v) \in [2^{j+1}, 2^{j+2}]$, since if $d(s,v) \in [2^r,2^{r+1}], r \leq j$, we can use the same argument for the ES tree on $G^r$. Let $L_{t,j+1}(s,v)$ be the level of $v$ in the monotone ES tree of $G^{j+1}$ maintained up to depth $\lceil \frac{2(2\beta+1)}{\epsilon_2} \rceil$. Our goal is to show,
\[ d_{t,j+1}(s,v):=\eta(2^{j+1}, \epsilon_2) L_{t,j+1}(s,v) \leq (1+\epsilon_{j+1})d_{t,G}(s,v)\]
 As discussed in Lemma \ref{lem:hop_doubling}, we consider the shortest path between $s$ and $v$ in $G$, and first divide it into two segments $\pi_1$ and $\pi_2$ each with length at most $2^{j+1}$.
Then divide each one of $\pi_1$ and $\pi_2$ into segments and consider the case by case inductive analysis as we did before for showing the stretch in $H_{j+1}$. 
We argue why the levels in tree rooted at $s$ corresponding to $\pi_1$ have the desired stretch, then a similar reasoning with a factor $2$ in the number of hops follows for $\pi_2$. 

We use a case-by-case analysis similar to what we used for showing properties of $\bar{H}_{j+1}$, and consider the paths that were inductively constructed for each segment $\bar{H}_{j+1}$. Using that structure, we argue that in the monotone ES tree on $G^{j+1}$ we can maintain the levels such that for each $0 \leq i \leq 1/\rho+k+1$ one of the following conditions holds:
\begin{enumerate}
    \item
    We have $d_{t,j+1}(s,v) \leq \eta(2^j, \epsilon_2) L(s,v) \leq (1+ \epsilon_{j+1}) d_{t,G}(s,v)$, where this estimate corresponds to a path with $\beta_i =(3/\delta)^{i}$ in $H_{j+1}$ (and hence $G^{j+1}$).

    \item There exists $z_1 \in A_{i+1}$, such that $d_{t,{j+1}} (s,z_1) \leq \eta(2^{j+1}, \epsilon_2) L(s,z_1) \leq 2(1+\epsilon_j)(1+\epsilon_2) d_{t,G}(s,v)$ that corresponds a path with $\beta_i =(3/\delta)^{i}$ hops.
\end{enumerate}

Then we can use this to show that after all iterations there either exists a path of depth at most $\lceil \frac{2(2\beta+1)}{\epsilon_2} \rceil$ on $G^{j+1}$ between $s$ and $v$ with stretch $(1+\epsilon_{j+1})$, or the monotone ES tree returns and estimate with this stretch.

We briefly review the different cases, same as before. First assume that we have $s \in B(v)$ for some iteration $1 \leq i \leq 1/\rho +1 +k$, which implies we have added a hopset edge with weight $w_{j+1}$ to $\bar{H}_{j+1}$. In this case the edge $(s,v)$ was directly added to $H_{j+1}$. If edge $(s,v)$ is stretched then we set $ L_{t,G}(s,v)=L_{t-1,G}(s,v)$, and by induction on time we have \[\eta(2^{j+1}, \epsilon_2) L_{t,j+1}(s,v)= \eta(2^{j+1}, \epsilon_2) L_{t-1,j+1}(s,v) \leq (1+\epsilon_{j+1}) d_{t,G}(s,v).\] 
If this edge is not stretched then by Lemma \ref{lem:rounding} after scaling we get distance at most $(1+\epsilon_{j})(1+\epsilon_2)^2d_{t,G}(s,v)$, where the additional factor of $(1+\epsilon_2)$ is due to scaling of $G \cup \bar{H}_{j+1}$.

Now consider the case $s \not \in B(v)$. Recall that we inductively showed one of the two theorem conditions hold for each $i$ for length in $H_{j+1}$, and we now argue that this corresponds to one of the two conditions above for the same $i$, but now on $G^{j+1}$. Let $\pi_i$ be the path in $H_{j+1}$ that satisfies one the theorem conditions for a fixed $i$.

First assume that no edge on this path is stretched. Then the stretch argument for $L(s,v)$ clearly holds based on the earlier arguments and Lemma \ref{lem:rounding}. Now let us argue about the possible insertions on this path, i.e.~when an edge added on $\pi_i$ is strecthed with respect to $s$. Note that by our construction, and in all cases we considered in our hopset argument, an edge $(x',y')$ was inserted into $H_{j+1}$ only when $x' \in B(y')$ for some $0 \leq i \leq k+1/\rho+1$, and the weights were assigned based on Observation \ref{obs:monotone}. Using these weights, we prove a claim that allows us to reason about possible insertions on $\pi_i$. At a high level, we show that the level of $y'$ is either determined by an estimate at time $t-1$ for $d(s,y')$ or by the level of a node $x'$, and a single edge between $(x',y')$ with weights satisfying Observation \ref{obs:monotone}. In other words, in the second case using a case by case analysis same as before, we know that for any node $y'$ there exists another node, in this case $x'$, that shortcuts the path from $s$ to $y'$ using one edge. 

\begin{claim} \label{claim:sssp_inserts}
Let $(x',y')$ be an edge added to $H_{j+1}$ and hence $G^{j+1}$ with weight $w_{G^{j+1}}(x',y')$ due to the fact that $x' \in B(y')$. Then either of the following holds for the level of node $y'$ in the monotone ES tree rooted at $s$:
\begin{itemize}
    \item $L_{t,j+1}(s,y') = L_{t-1, j+1}(s,y')$ and thus $\eta(2^{j+1},\epsilon_2) L_{t,j+1}(s,y') \leq \eta(2^{j+1},\epsilon_2) L_{t-1,j+1}(s,y') \leq (1+\epsilon_{j+1})d_{t,G}(s,y')$; or,
    \item We have $L_{t,j+1}(s,y') \leq L_{t,j+1}(s,x') + w_{G^{j+1}}(x',y')$.
\end{itemize}
\end{claim}
\begin{proof}

 The first case is when the edge $(x',y')$ is stretched in the tree rooted at $s$ on $G^{j+1}$. Note that this is different from the setting in Observation \ref{obs:monotone}, where we were reasoning about the node $y'$ being stretched in the tree rooted at $x'$ on $G^j$. In this case we set $L_{t,j+1}(s,y')=L_{t-1,j+1}(s,y')$. Since we have maintained distances up to depth $\lceil \frac{2(2\beta+1)}{\epsilon_2}\rceil$ on $G^{j+1}$ with stretch $(1+\epsilon_j)$ at time $t-1$, and since we are in the decremental setting this means that after scaling back we get the desired stretch.
 
 The second case is when the edge $(x',y')$ is not stretched in the tree rooted at $s$. The claim follows by definition of an edge that is not stretched.  
\end{proof}

 Note that if $d_{t,G}(x',y')$ belonged to a smaller scale, we have already added an edge that satisfied a similar condition for the corresponding scale.

Going back to the hopset argument, we note that every insertion into $H_{j+1}$ on path $\pi_i$ (and edge that is stretched with respect to $s$) satisfies the conditions in Claim \ref{claim:sssp_inserts}. In other words, for any node on the path, say $y'$, there exists a node $x'$ that is directly connected to $y'$, satisfying the stretch in Claim \ref{claim:sssp_inserts}.
 This combined with what we proved inductively on the structure of segments of path $\pi_1$ in $H_{j+1}$, implies that for any node $v, d(s,v) \leq 2^{j+1}$ we have a path with the desired stretch that is consisted of all the edges added for different $i$. Finally, after the scaling we can obtain the desired stretch in which we lose another factor of $(1+\epsilon_2)$.

We briefly review the cases that, at a high-level, shows that a such a node $x'$, satisfying Claim \ref{claim:sssp_inserts} exists that appropriately \textit{shortcuts} the distance from $s$ to $y'$ for any node $y'$ on $\pi_i$ stretched with respect to $s$.

Recall the hopset argument for $i$: an insertion into one of the segments (of length $\delta d(s,v)$) can only occur when condition two of the theorem is satisfied for some node in $z \in A_{i+1}$. Let $[z'_l,z'_r]$ be the segment for which the new edge was inserted.  
 We argued that either $z'_r \in B(z'_l)$ or there is another node $z'$ for which the second theorem condition holds and $z' \in B(z'_l)$.

In any case we inserted a single edge in $H_{j+1}$ on this segment with weights satisfying Observation \ref{obs:monotone}. Then using Claim \ref{claim:sssp_inserts}, and similar calculation as we did for $H_{j+1}$ we can show that the second condition also holds on $G^{j+1}$, but there is an additional error factor of $(1+\epsilon_2)$ from scaling.

 At a high level, we have shown that the inserted edge on $\pi_i$ has a length that appropriately shortcuts the last segment, otherwise no new edges were added in iteration $i$ (when the first theorem condition holds for all segment).

We argued earlier that this path $\pi_i$ has stretch $(1+\epsilon_j)(1+\epsilon)(1+\epsilon_2)$ in $G \cup \bar{H}_{j+1}$. Hence after scaling and running Algorithm \ref{alg:estree} on $G^{j+1}$, we know that path $\pi_i$ has depth $ \lceil \frac{2(2\beta+1)}{\epsilon} \rceil$ and we have the following estimate for $v$:
\[d_{t,j+1}(s,v)\leq \min_{r=1}^{j+1} \eta(2^{j+1}, \epsilon_2) L_{t, r}(s,v) \leq (1+\epsilon_j)(1+\epsilon_2)^2(1+\epsilon) d_{t,G}(s,v)
\] 
Then after all the iterations $1 \leq i \leq 1/\rho +1 +k$, the second condition cannot hold (since $A_{1/\rho +1 +k}=\emptyset$), the first condition must hold, which states that there
is a path with  $\beta =(3/\delta)^{1/\rho +1 +k}$-hops and stretch $(1+\epsilon_{j+1})d_{t,G}(s,v)$ in $G \cup \bar{H}_{j+1}$ between $s$ and $v$. Also by path doubling of Lemma \ref{lem:hop_doubling} we argued that this also means that there is a path with $2\beta+1$ hops and $(1+\epsilon_{j})(1+\epsilon_2)(1+\epsilon)$-stretch in $G \cup \bar{H}_{j+1}$ between $s$ and $v$ that is consisted of two paths satisfying the first theorem condition for $H_{j+1}$, and that this corresponds to a path with stretch $1+\epsilon_{j+1}$ in $G^{j+1}$.
The concatenation of this same paths in $G^{j+1}$ approximates $\pi_i$ and after scaling and unscaling we will have an additional factor of $(1+\epsilon_2)$.

\end{proof} 

Theorem \ref{thm:single_stretch} allows us to hierarchically use the restricted hopsets for smaller scales to compute the distance for larger scales, that is in turn used to update the hopset edges in the larger scales.   

In the following lemma, we will show that by setting $\delta=O(\frac{\epsilon}{(k+1/\rho +1)})$ we get the desired stretch for Lemma \ref{lem:single-scale}. Next, we use Lemma \ref{lem:single-scale} for all scales and by setting the appropriate error parameters we can prove our overall stretch and hopbound tradeoffs. We also prove the overall update time using the running time of the monotone ES tree algorithm to run the restricted hopset algorithm on the hopsets obtained for each scale.

\paragraph{Single scale stretch.} We will now use the stretch argument above to get the hopbound and stretch for each scale by setting the appropriate parameters.
As discussed, there are \textit{two} error factors incurred in each scale. One is caused by the fact that we are using previously added hopset edges, which we denoted by $(1+\epsilon_j)$ for scale $j$, and another is caused due to the rounding error, which we denote by $(1+\epsilon_2)$. To get an overall stretch of $(1+\epsilon)$, we will set $\epsilon'= \frac{\epsilon}{6 \log W}$ and $\epsilon_2=\epsilon'$.

\begin{corollary}\label{cor:overall_stretch}
After each update $t$, and for all $j, 0 \leq j \leq \log W$ and any pair $x,y \in V$, where $2^j \leq d_{t,G}(x,y) \leq 2^{j+1}$, we have $d_{t,G}(x,y) \leq d_{t,G \cup \bar{H}_j}(x,y) \leq (1+3\epsilon')^j \cdot d^{(\beta)}_{t,G}(x,y)$. 
\end{corollary}
\begin{proof}
We use an induction on $j$.  The base case ($j = 0$) is satisfied by the paths in $G$, since we can assume with out loss of generality that the edge weights are at least one.
First, by induction hypothesis, we have a $(2^{j}, \beta, (1+3\epsilon')^{j})$-hopset, and hence $1+\epsilon_j = (1+3\epsilon')^{j}$ .

We then use Theorem \ref{thm:single_stretch}, for $\epsilon_2=\epsilon'$, and $\delta= \frac{\epsilon'}{8(k+1/\rho +1)}$. For the final iteration $i=\frac{1}{k+1/\rho +1}$ since $A_{i+1} = \emptyset$, the second item can not hold. Hence the first item should hold, and since $8\delta i < \epsilon'$ we have,
\[    d^{(\beta)}_{t,G \cup \bar{H}_j}(x,y) \leq  (1+3\epsilon')^{j-1}(1 +\epsilon')(1+\epsilon') d_{t,G}(x,y)  \leq (1+3\epsilon')^j d_{t,G}(x,y). \]

Here $d_{t,j}(x,y)$ is the sum of weights in the monotone ES tree, which corresponds to the approximate $\beta$-limited distance of $x$ and $y$ on the scaled graph.

\end{proof}
\paragraph{Proof of stretch and hopbound in Lemma \ref{lem:single-scale}.}
Now by simply setting $\epsilon'= \frac{\epsilon}{3}$ in Corollary \ref{cor:overall_stretch}  we get the desired stretch and hopbound.
\paragraph{Putting it together.}
We now use the stretch argument of Corollary \ref{cor:overall_stretch} with the update time followed by Lemma \ref{lem:single-scale} to get the following hopset guarantees. 
\begin{theorem}\label{thm:main_hopset}
The total update time in each scaled graph $G^j$, $1 \leq j \leq \log W$, over all deletions is $\tilde{O}((\ell/\epsilon') (n^{1+\nu}+ m)n^{\rho})$, and hence the total update time\footnote{If weights are not polynomial the $\log n$ factor will be replace with $\log W$, and a factor of $\log^2 W$ will be added to the update time.} for maintaining $(\beta, 1+\epsilon)$-hopset with hopbound $\beta= O(\frac{\log n}{\epsilon} \cdot (k+1/\rho))^{k+1/\rho+1}$ is $\tilde{O}(\frac{\beta}{\epsilon} \cdot mn^{\rho})$.
\end{theorem}
\begin{proof}
 First we use Corollary \ref{cor:overall_stretch} to prove the stretch and hopbound, by setting $j= \log (nW)$. For the final scale we have $d_{t,\log nW}(u,v)= (1+3\epsilon')^{\log (nW)} d_G(u,v) \leq (1+\epsilon) \log (nW)$. The hopbound obtained is 
 \[O(\frac{1}{\epsilon'} \cdot (k+1/\rho))^{k+1/\rho+1}=O(\frac{\log (nW)}{\epsilon} \cdot (k+1/\rho))^{k+1/\rho+1}.\]

 The running time follows by Lemma \ref{lem:single-scale} where $\Delta= O(n^{1+\nu})$, we get an overall running time of $\tilde{O}(mn^{\rho} \cdot \frac{\beta}{\epsilon})$ time.
\end{proof}
Hence for constant $\rho =1/k$ the total update time is $\tilde{O}(mn^{\rho})$ and the hopbound $\beta$ is polylogarithmic.

\section{Applications}\label{sec:applications}
In this section we explain two applications of our decremental hopsets to get improved bounds for $(1+\epsilon)$-approximate SSSP and MSSP and $(2k-1)(1+\epsilon)$-APSP. For both of these problems we first construct a hopset, where we choose the appropriate hopbound depending on the number of source. We then use the scaling scheme in Lemma \ref{lem:rounding} on the obtained graph.

Our algorithm for $(2k-1)(1+\epsilon)$-APSP involves maintaining two data structures simultaneously: A $(\beta, 1+\epsilon)$-hopset, and a Thorup-Zwick distance oracle \cite{TZ2005}. At a high-level the hopset will let us maintain the distance oracle much faster, at the expense of a $(1+\epsilon)$-factor loss in the stretch. 

\subsection{$(1+\epsilon)$-approximate SSSP and  $(1+\epsilon)$-MSSP}
Given a graph $G=(V,E)$ and a set $S$ of size of sources, our goal is to maintain the distance from each source in $\tilde{O}(sm+mn^\rho)$, total update time (where $\rho$ is a constant), and constant query time.

Once a $(\beta,\epsilon)$-hopset is constructed, we can run Algorithm \ref{alg:estree} on all the scaled graphs  $G^1, G^2, ..., G^{\log (nW)}$ up to depth $O(\beta)$, scale back the distances, and return the smallest value to each source. 

In the next theorems we argue that using the same techniques as we used for maintaining the hopset (that are similar to framework of \cite{henzinger2016}), namely by combining monotone ES tree and scaling, we get our SSSP and MSSP results. In particular after constructing the hopset we can use Theorem \ref{thm:single_stretch} and Theorem \ref{thm:main_hopset} to get,

\begin{theorem}\label{thm:mssp}
Given an undirected and weighted graph $G=(V, E)$, there is a decremental algorithm for maintaining $(1+\epsilon)$-approximate distances from a set $S$ of sources in total update time of $\tilde{O}(\beta (|S| (m+n^{1+\frac{1}{2^k-1}})+  mn^\rho))$, where $\beta= O(\frac{\log (nW)}{\epsilon} \cdot (k+1/\rho)^{k+1/\rho+1})$, and with $O(1)$ query time.   
\end{theorem}
\begin{proof}
We maintain a $(\beta, \frac{\epsilon}{3})$-hopset $H$ based on Theorem \ref{thm:main_hopset}. Then we run Algorithm \ref{alg:estree} on $G \cup H$ from all the $s$ for all scaled graphs. The the claim follows by the argument in Theorem \ref{thm:single_stretch}. 
In particular, after adding all the hopset edges at time $t$ for all scales, we will run the monotone ES tree algorithm rooted at each source again on the union of all scaled graphs $G^1 \cup ...\cup G^{\log W}$ (by setting $\epsilon_0=\epsilon/3$) and let the level $L(s,v)$ of a node be $\min_{j} \eta( 2^j,  \frac{\epsilon}{3}) L_j(s,v)$ where $L_j(s,v)$ is the level of $v$ on $G^j$ after running the monotone ES tree that is run up to depth $\beta$. By item 3 of Theorem \ref{thm:main_hopset}, we get an overall stretch of $(1+\epsilon/3)^2 \leq (1+\epsilon)$. 


The time required for maintaining the hopset is $\tilde{O}( (m+ n^{1+\frac{1}{2^k-1}}) n^\rho)$ and by setting $n^{\rho}=s$ the time required for maintaining $\beta$-hop bounded shortest path from all sources is $O(sm \cdot \beta) = \tilde{O}(sm)$, when $s=n^{\Omega(1)}$.
\end{proof}

We next state two specific consequences. First implication is that when the number of sources is a polynomial, and the graph is not very sparse, we can get a near-optimal (up to polylogarithmic factors) algorithm for $(1+\epsilon)$-MSSP.
\begin{corollary}
Given an undirected and weighted graph $G=(V, E)$, where $|E|= n^{1+\Omega(1)}$, there is a decremental algorithm for maintaining $(1+\epsilon)$-approximate distances from $s$ sources, where $s=n^{\Omega(1)}$ in total update time of $\tilde{O}(sm)$, and with $O(1)$ query time.   
\end{corollary}

When the number of sources $s= n^{o(1)}$ (e.g.~in case of SSSP), the best tradeoff can be obtain by setting $\rho= \frac{\log \log n}{\sqrt{\log n}}$. We will then have $\beta= 2^{\tilde{O}(\sqrt{\log n})}$ and also $n^{\rho}=2^{\tilde{O}(\sqrt{\log n})}$. In this case we get improved bounds over the result of \cite{henzinger2014}, which has a total update time of is $mn^{\tilde{O}({\log^{3/4}n})}$. 
\begin{corollary}
Given an undirected and weighted graph $G=(V, E)$, there is a decremental algorithm for maintaining $(1+\epsilon)$-approximate distances from $s$ sources, when $0 <\epsilon<1$ is a constant and $|E|= n \cdot 2^{\tilde{\Omega}(\sqrt{\log n})}$, with total update time of  $\tilde{O}(sm \cdot 2^{\tilde{O}(\sqrt{\log n})})$, and with $O(1)$ query time.  Hence, we can maintain $(1+\epsilon)$-approximate SSSP in $ 2^{\tilde{O}(\sqrt{\log n})}$ amortized time.
\end{corollary}

\subsection{APSP distance oracles} \label{sec:APSP}
It is known that in static settings for any weighted graph $G=(V,E)$, we can construct a Thorup-Zwick \cite{TZ2005} distance oracle of size (w.h.p.) $\tilde{O}(n^{1+1/k})$, such that after the preprocessing time of $\tilde{O}(mn^{1/k})$, we can query $(2k-1)$-approximate distances for any pair of nodes in $O(k)$ time. In this section we show that in decremental settings we can maintain these distance oracles in total update time of $\tilde{O}(mn^{1/k})$ (for graphs that are not too sparse), and we can query $(2k-1)(1+\epsilon)$-approximate distances in $O(k)$ time. This can be done by maintaining a $(\beta, 1+\epsilon)$-hopset and a distance oracles for $G$ at the same time, where $\beta$ is polylogarithmic in $n$. Intuitively, the hopset will allow us to update distances faster on the distance oracles. 

\paragraph{Distance oracle algorithm via a hopset.} 
Assume that we are given a $(\beta, 1+\epsilon)$-hopset for $G$. The algorithm for constructing the Thorup-Zwick distance oracle is as follows: Similar to the algorithm in Section \ref{sec:static_hopset}, we define sets $V=A_0 \supseteq A_1 \supseteq ... \supseteq A_k=\emptyset$\footnote{This $k$ should not be confused with the size parameter in the hopset algorithm of Section \ref{sec:static_hopset}. Here we only use the fact that the hopset size can be bounded based on the graph density.}. But here each set $A_{i+1}$ is obtained by sampling each element from $A_i$ with probability $p_i=n^{-1/k}$. Same as before, for every node $u \in A_i\setminus A_{i+1}$, let $p_i(u) \in A_{i+1}$ be the closest node to set $A_{i+1}$. We the bunch of a node $u$ is the set $B(u) = \cup_{i=1}^{k} B_i(u)= \{ v \in A_i: d(u,v) < d(u,A_{i+1})\} \cup \{p(u)\}$, $C(v)$ called the cluster of $v$ such that if $v \in B(u)$ then $u \in C(v)$. The distance oracle is consisted of bunches $B(v)$ for all $v \in V$, and the distances associated with them. Note that the information stored here are also different from the hopset algorithm described in Section \ref{sec:static_hopset}, since there we only added edges for nodes $v \in A_i$ and their bunches. 
Thorup and Zwick \cite{TZ2005} show that this distance oracle has the following properties (in static settings):
\begin{theorem}[\cite{TZ2005}]
There is a distance oracle of expected size $O(kn^{1+1/k})$, that can answer $(2k-1)$-approximate distance queries for a given weighted and undirected graph $G=(V,E)$ in $O(k)$ time for any $k \geq 2$. The preprocessing time in static settings is w.h.p.~$\tilde{O}(mn^{1/k})$.
\end{theorem}

As discussed in Appendix \ref{app:restricted_hopset}, Roditty and Zwick \cite{roditty2004} showed how to maintain this data structure in $O(mn)$ update time for \textit{unweighted graphs}, but where the size is increased to $\tilde{O}(m+n^{1+1/k})$. For weighted graphs their updates time can be as large as $O(mn^{1+1/k})$. 
We will argue that by maintaining a $(\beta,1+\epsilon)$-hopset along with the distance oracle we can improve the total update time to $\tilde{O}(\beta mn^{1/k})$. This combined with our decremental hopset of Theorem \ref{thm:main_hopset} will lead to the desired bounds. More formally,
\begin{theorem}\label{thm:oracle_time}
Given a weighted and undirected graph $G=(V,E)$ and a $(\beta, 1+\epsilon)$-hopset $H$ for $G$, and a parameter $k \geq 2$, we can maintain a distance oracle with size $\tilde{O}(m+ |E(H)|+n^{1+1/k})$ that supports $(1+\epsilon)(2k-1)$-approximate queries in $\tilde{O}(\frac{\beta}{\epsilon} \cdot m n^{1/k})$ total update time with $O(k)$ query time.
\end{theorem}
\begin{proof}
Similar to Theorem \ref{thm:mssp}, we consider the sequence $G^1, ..., G^j$, where $G^r, r \leq j$ is scaling of the graph $G \cup \bar{H}_r$ as defined in Section \ref{sec:new_hopset} (and Algorithm \ref{alg:main}), where $\epsilon_0= \frac{\epsilon}{3}$ and $\bar{H}_j$ is a $(2^j, \beta, \frac{\epsilon}{3})$ hopset. We then run the algorithm of Roditty-Zwick \cite{roditty2004} on $G$ up to depth $\lceil 3\beta/\epsilon \rceil$ for maintaining the clusters and the bunches.
The algorithm and the running time analysis is similar to the restricted hopset algorithm described in Section \ref{sec:new_hopset}. The main differences in these algorithms are the sampling probabilities and the information stored. Therefore using the argument in Lemma \ref{lem:dec_bound_cluster} we can show that by running this algorithm on $G$ with depth $\lceil 3\beta/\epsilon \rceil$ we can maintain a bunch $B_i(u)$ for all nodes $u \in V, 1 \leq i \leq k-1$ in $\tilde{O}(\frac{\ell m}{\epsilon q_i})=\tilde{O}(\frac{\beta}{\epsilon} mn^{1/k})$ total update time. This algorithm lets us maintain clusters. We also maintain the distances in clusters and hence bunches as follows: For each $v \in V, u \in B(v)$, we run single-source shortest path distance between from $v$ on scaled graphs $G^1, ..., G^{\log W}$ (by setting $\epsilon_0=\epsilon/3$). We then set the distance  $d(u,v)$ to be $\min_{j} \eta( 2^j,  \frac{\epsilon}{3}) L_j(s,v)$ where $L_j(s,v)$ is the level of $v$ on $G^j$ after running the monotone ES tree that is run up to depth $\lceil 6(\beta+1)/\epsilon)\rceil$.

Again, when we combine the hopset stretch with the stretch with the rounding algorithm caused by rounding, we get an overall stretch of $(1+\frac{\epsilon}{3})^2 \leq 1+\epsilon$.  The overall stretch is thus $(2k-1)(1+\epsilon)$.
\end{proof}

\begin{theorem} \label{thm:oracle_main}
Given a weighted graph $G=(V,E)$ with polynomial weights, and constant\footnote{If $k=\omega(1)$, then a factor of $n^{o(1)}$ will be added to the running time.} $k \geq 2$ and $0 < \epsilon <1$, we can maintain a data structure with expected size $\tilde{O}(m+n^{1+1/k})$ and total update time of $\tilde{O}(mn^{1/k}\cdot (1/\epsilon)^{O(1)})$, that returns $(2k-1)(1+\epsilon)$-stretch queries for any pair $u,v \in V$ with $O(k)$ query time. 
\end{theorem}
\begin{proof}
We construct and maintain a $(\beta, 1+\frac{\epsilon}{3})$-hopset using Theorem \ref{thm:main_hopset}. If $m= n^{1+\Omega(1)}$ we can set $\rho=\frac{1}{k}$, and we set the hopset size parameter $\nu$ to a small constant\footnote{The choice of size parameter impact the polylogarithmic factors. Hence one option is to choose the smallest constant such that the graph size is not smaller than the hopset size.} such that $O(n^{1+\nu})=O(m)$ ($\nu =\frac{1}{2^k-1}$). If $m=n^{1+o(1)}$, we set $\rho=\frac{1}{2k})$. In both cases time required for maintaining a hopset is $\tilde{O}(mn^{1/k} \cdot (1/\epsilon)^{O(1)})$. We get hopbound $\beta= O(\log n/\epsilon)^{\log (1/\nu) +1/k+1}= \textit{polylog }(n)$. Hence we can also maintain the distance oracle in $\tilde{O}(mn^{1/k})$ total update time. The stretch will be $(2k-1)(1+\epsilon)$, and the query time remains the same as the static query time, which is $O(k)$.
\end{proof}

\paragraph{Distance Sketches.} As shown in \cite{TZ2005}, the bunch $B(v)$ of each node $v$ can also be seen as a distance sketch of expected size $O(kn^{1/k})$. They show that only using $B(u)$ and $B(v)$ we can approximate the distance of $u$ and $v$. Since we are maintaining distances using a hopset, we will lose an additional factor of $(1+\epsilon)$ in the stretch. Note that we need access to $G$ in order to \textit{update} the sketches, but the rest of the graph is not needed in order to \textit{query} the distances. The update time is the same as Theorem \ref{thm:oracle_main}, so we do not repeat the statement. 
\appendix

\bibliographystyle{alpha}
\bibliography{dynamic_refs}
\section{Static Hopset Properties} \label{app:static_hopset}
In this section, we will briefly overview the (static) hopset algorithm of \cite{elkin2019RNC} (which is similar to \cite{huang2019}). Given a weighted graph $G=(V,E)$, and a parameter $k$, we first construct a hopset of size $O(n^{1+\frac{1}{{2^k}-1}})$ and hopbound $O(k/\epsilon)^{k}$. We then modify the algorithm in order to get a better running time at the cost of a worse stretch.	 
 We define sets $V=A_0 \supseteq A_1 \supseteq ... \supseteq A_k=\emptyset$. Let $\nu =\frac{1}{2^k-1}$. Each set $A_{i+1}$ is obtained by sampling each element from $A_i$ with probability $q_i=n^{-2^i \cdot \nu}$. Hence it can be shown that $E[|A_i|]= n^{1-2^{i-1}\nu}$. 
 
 For every vertex $u \in A_i\setminus A_{i+1}$, let $p(u) \in A_{i+1}$ be the closest node to set $A_{i+1}$. The bunch is set to $B(u)= \{ v \in A_i: d(u,v) < d(u,A_{i+1})\} \cup \{p(u)\}$. Also, we define a set $C(v)$ called the cluster of $v$ such that if $v \in B(u)$ then $u \in C(v)$. It can easily be shown that clusters are \textit{connected} in a sense that if a node $v \in C(u)$ then any node $z$ on the shortest path between $v$ and $u$ is also in $C(u)$. As we will see, this property is important for bounding the running time. The hopset is consisted of adding edges $(u,v)$ where $v \in B(u)$, and setting the weight to be $d(u,v)$.

 The algorithm described will lead to a $((k/\epsilon)^k, \epsilon)$-hopset, but the running time can be as large $O(mn)$. In order to resolve this, \cite{elkin2019RNC} proposed an algorithm using modified sampling probabilities of $q_i=\max(n^{-2^i \cdot \nu}, n^{-\rho})$.
Using this approach, the number of iterations becomes $k+1/\rho+1$, but the hopbound is also increase to $O(\frac{k+1/\rho+1}{\epsilon})^{O(k+1/\rho+1)}$.  
We briefly review some the properties of this hopset algorithm as discussed in \cite{elkin2019RNC, huang2019}, and will then explain how \cite{elkin2019RNC} modifies the algorithm to improve the running time. 

One important component of this algorithm is the \textit{modified Dijsktra's algorithm} that we will also utilize in our dynamic algorithms, and thus we breifly review it. This algorithm was presented by Thorup-Zwick \cite{TZ2005} and it allows us to construct the bunches and clusters for level $i$ in $O((m + n\log n)/q_i)$ (expected) time. 
At a high-level this is done by making a modification to Dijkstra. In the original Dijkstra for each source $u \in A_i \setminus A_{i+1}$, at each iteration we consider an unvisited vertex $v$, and relax each incident edge $(v,z)$ by setting ${d}(u,z) := \min \{d(u,z), d(u,v)+ w(v,z)\}$. But in the modified algorithm this is done only if $d(u,v) + w(v,z) < d(z, A_{i+1})$. In other words each node $z$ only ``participates" in a shortest-path exploration from a source $u$ only if $z \in B(u)$. Note that if $z \in B(u)$, all the nodes on the shortest path between $u$ and $z$ are considered. This will let us bound the running time by expected size of $B(u)$.
\subsection{Properties} 

\paragraph{Size.} We rely on the fact that the hopset size is not denser than the original graph. This is why our main bounds do not hold for very sparse graphs. For analyzing the size, \cite{elkin2019RNC} argues that for each $u \in A_i \setminus A_{i+1}$ we have $E[|B(u)|] \leq 1/q_i$ for the following reason: Consider an ordering of vertices in $A_i$ based on their distance to $u$. By definition, size of $B(u)$ is bounded by the number of vertices in this ordering until the first vertex in $A_{i+1}$ is visited. This corresponds to a geometric random variable with parameter $q_i$ and thus in expectation it is $1/q_i =n^{2^i \nu}$. Hence for all $i$ the number of edges added is in expectation
\[\sum_{i=1}^{k-2} E[|A_i|]n^{2^i \cdot \nu } = O(kn^{1+\nu}). \]

\paragraph*{Modified Dijsktra's algorithm.}  For an efficient construction of these hopsets, \cite{elkin2019RNC} used the \textit{modified Dijsktra's algorithm}, which was proposed by Thorup-Zwick \cite{TZ2005}. This algorithm the bunches for level $i$ can be constructed in $O(m + n\log n)/q_i$. At a high-level this is done by making a modification to Dijkstra. In the original Dijkstra's, for each source $u \in A_i \setminus A_{i+1}$, at each iteration we consider an unvisited vertex $v$, and ``relax" each incident edge $(v,z)$ by setting ${d}(u,z) = \min \{d(u,z), d(u,v)+ w(v,z)\}$. But in the modified algorithm this is done only if $d(u,v) + w(v,z) < d(z, A_{i+1})$. In other words each node $z$ only ``participates" in a shortest-path exploration from a source $u$ only if $z \in B(u)$. Note that if $z \in B(u)$, all the nodes on the shortest path between $u$ and $z$ are considered. Since $|B(u)| \leq \frac{1}{q_i}$, this allows us to bound the running time by $O(m n^{\rho})$.

\subsection{Hopbound and Stretch.} \label{app:hopbound_stretch}
In this section we sketch the anlysis of the hopbound and stretch of a simple static hopset algorithm, that does not bound the sampling probabilities by $n^{-\rho}$. This leads to (almost) optimal size and hopbound tardeoff but has a larger construction time. 
The extension to the more efficient variant will be straightforward. 
The following lemma was proved by \cite{elkin2019RNC, huang2019}. We give a proof sketch here. We use a similar idea in our dynamic hopset construction (in combination with monotone ES tree and scaling), and hence some of the missing details can be found in proof of Theorem \ref{thm:single_stretch}.
\begin{lemma} \label{lem:hopbound}
Fix $0 < \delta \leq 1/{8k}$, and consider a pair $x,y \in V$. Then for $0 \leq i \leq k+1$ we have either of the following conditions:
\begin{itemize}
    \item $d^{((3/\delta)^i)}_{G \cup {H}}(x,y) \leq (1+8\delta i) d_{G}(x,y)$ or,
\item There exists $z \in A_{i+1}$ such that, 
$d^{((3/\delta)^i)}_{G \cup H} (x,z)\leq  2 d_{G}(x,y).$
\end{itemize} 
\end{lemma} 
\begin{proof}[Proof sketch]
This can be shown by an induction on $i$. For the base case of $i=0$, we have three cases. If $y \in B(x)$ then edge $(x,y)$ is in the hopset, and the first condition of the lemma holds. Otherwise if $x \in A_1$, then $z=x$ trivially satisfies the second condition. Otherwise we have $x \in A_0/A_1$, and by setting $z=p(x)$ we know that there is an edge $(x,z) \in H$ such that $d(x,z) \leq d(x,y)$ by definition of $p(x)$, and hence the second condition holds. 

Now assume the claim holds for $i$. Consider the shortest path $\pi(x,y)$ between $x$ and $y$. We divide this path into $1/\delta$ segments og length roughly $\delta d_G(x,y)$ (up to rounding). Using triangle inequality on all segments we then use the induction hypothesis on each segment. If for all the segments the first condition holds for $i$, then there is a path of $(3/\delta)^{i+1}$-hops consisted of the hopbounded path on each segment. We can show that this path satisfies the first condition for $i+1$.
 Now, assume that there are at least two segments for which the first condition does not hold for $i$. Then let $[u_\ell, v_\ell]$ be the first such segment (i.e.~closest to $x$) and let $[u_r, v_r]$ be the last such segment.
Then by inductive hypothesis there are $z_\ell, z_r \in A_{i+1}$ such that:
\begin{itemize}
\item  $d_{G \cup H}^{((3/\delta)^i)} (u_\ell, z_\ell) \leq 2d(u_\ell, v_\ell)$, \textit{and,}
\item $d_{G \cup H}^{((3/\delta)^i)} (v_r, z_r) \leq 2d(u_r, v_r)$
\end{itemize}
Again, we consider two cases. First, in case $z_r \in B(z_\ell)$, we have added a single hopset edge between $z_r$ and $z_\ell$ with weight $d(z_r, z_\ell)$. By applying the inductive hypothesis on segments before $[u_\ell, v_\ell]$, and after $[u_r,v_r]$, we have a path with at most $(3/\delta)^i$ for each of these segments, satisfying the first condition for the endpoints of the segment. Also, we have a $2(3/\delta)^i +1$-hop path going through $u_{\ell}, z_{\ell}, z_r, v_r$ that satisfies the first condition for $u_{\ell}, v_r$. Putting all of these together, we can show that there is a path of hopbound $(3/\delta)^{i+1}$ satisfying the first condition. To get this we need to use the fact that the length of each segment is at most $\delta \cdot d(x,y)$. We have,
\begin{align*}
    d^{(3/ \delta)^{(i+1})}_{G \cup H} (x,y) &\leq \sum^{\ell-1}_{j=1} [d_{G \cup H}^{((3/\delta)^i)} (u_j, v_j) +d^{(1)}_G(v_j,u_{j+1})] +d^{((3/\delta)^i)}_{G \cup H} (u_\ell, z_\ell)\\
    &+d^{(1)}_H (z_\ell, z_r) + d^{((3/\delta)^i)}_{G \cup H} (u_r, v_r)  +d^{(1)}_H (v_r, u_{r+1})\\
    &+\sum^{(1/\delta)}_{j=r+1} [d_{G \cup H}^{((3/\delta)^i)} (u_j, v_j) +d^{(1)}_{G}(v_j,u_{j+1})]\\
    &\leq 8 \delta d_G(x,y) + (1+ 8 \delta i) d_G(x,y)\\
    &\leq (1+ 8 \delta(i+1)) d_G(x,y)
\end{align*}

Finally, consider the case where $z_r \not \in B(z_{\ell})$. If $z_{\ell} \not \in A_{i+2}$, we consider $z =p(z_\ell)$. By definition we have added the edge $(z_{\ell}, z)$ to the hopset, we can show that the second condition holds. We use similar reasoning as before and also use the fact that we set $\delta <1/{8k}$ to show that item 2 holds in this case. The only remaining case is when $z_{\ell} \in A_{i+2}$, in a similar but simpler reasoning follows by setting $z=z_{\ell}$.
\end{proof} 

We can now set $\delta =\Theta(k/\epsilon)$ in Lemma \ref{lem:hopbound}, and since $A_{k}=\emptyset$, for $i= k-1$ only the first condition can hold. Therefore we get a hopbound of $\beta= \Theta(k/\epsilon)^{k}$.

\paragraph{Hopbound in the efficient variant.} For more efficient construction, we considered a two phase algorithm. For the first phase we use similar reasoning as Lemma \ref{lem:hopbound}, but in the second phase the parameters change. The algorithm will require more \textit{iterations} that will impact the overall hopbound. We require $k+1/\rho+1$ iterations overall. In the second phase we have $\delta'= (k+1/\rho)/\epsilon$ and thus we have overall hopbound of $O(\frac{k+1/\rho}{\epsilon})^{k +1/\rho+1}$. 

By putting everything we have the following guarantees for the static hopset:
\begin{theorem}[\cite{elkin2019RNC}]
There is an algorithm that given a weighted and undirected graph $G=(V,E)$, and $2 \leq k \leq \log \log n -2$, $\frac{2}{2^k-1} < \rho <1$ computes a $(\beta, \epsilon)$-hopset of size $O(n^{1+\frac{1}{2^k-1}})$, where $\beta= O( ( \frac{k+1/\rho}{\epsilon}))^{k+1/\rho+1}$. It runs in $O(\frac{n^\rho}{\rho}( m +n\log n))$ expected time.
\end{theorem}

\section{Details Omitted from Section \ref{sec:restricted_hopset}} \label{app:restricted_hopset}
In this section, we review the restricted hopset algorithm that is mainly based on algorithm of \cite{roditty2004}, and alayze the running time.

\subsection{Algorithm of \cite{roditty2004}}
In this section, we review the algorithm of \cite{roditty2004}, which allows us to maintain restricted hopset as stated (but with additional property of handling certain edge insertions) can also be found in Algorithm \ref{alg:restricted_hopset} in Section \ref{sec:new_hopset}.

We sample sets $V=A_0 \supseteq A_1 \supseteq ... \supseteq A_{i(\rho)}=\emptyset$, where $i(\rho)= k+1/\rho+1$ once and they remain the same during the updates.
Next, we need to maintain values $d(v, A_i), 1 \leq i \leq k-1$ for all nodes $v \in V$. This can be performed by computing a shortest path tree rooted at a dummy node $s_i$ connected to all nodes in $A_i$. Let $\hat{d}=(1+\epsilon)d$. We can use the Even-Schiloach \cite{ES} algorithm up to depth $\hat{d}$ to compute all these distances in $O(\hat{d}m)$ time. The pivots $p(v), \forall v \in V$ can also be maintained in this process.

 \paragraph{Maintaining the clusters.} Recall that for $z \in A_i\setminus A_{i+1}$ we have $v \in C(z)$ if and only if $d(z,v) < d(v,A_{i+1})$. After each deletion, for each node $v$ and the cluster centers $z$ we first check whether the distance $d(z,v)$ has increased.  If $d(z,v) \geq d(v, A_{i+1})$, $v$ will be removed from $C(z)$. The more subtle part is adding nodes to new clusters. For each $0 \leq i <k$, we define a set $X_{i}$ consisted of all vertices whose distance to $A_i$ is increased as a result of a deletion, but where this distance is still at most $\hat{d}$. The sets $X_{i}$ can be computed while maintaining $d(v,A_i)$.
 
 Note that a node $v$ would join $C(w)$ only after an increase in $d(v,A_{i+1})$. Using this observation, after each deletion for every $v \in X_{i+1}, z \in B_i(u) \setminus B_i(v)$, and each edge $(u,v) \in E$ we check if $d(z,u) +w(u,v) < d(v, A_{i+1})$. If yes, then $v$ joins $C(z)$. We push $v$ to a priority queue $Q(z)$ with key $d(z,u)+w(u,v)$. If $v$ was already in the queue the key will be updated if this distance is smaller than the existing estimate. In this case we \textit{mark} $v$. The marked nodes join clusters $z$, but there may be other nodes that also need to join $C(z)$ as a result of this change.
 
 Hence after this initial phase, for each $z \in A_i \setminus A_{i+1}$ where $Q(z) \neq \emptyset$, we run the modified Dijkstra's algorithm. Recall that in the modified Dijkstra's algorithm when we explore neighbors of a node $x$, we only relax an edge $(x,y)$ if $d(x,y)+w(x,y) < d(x,A_i)$. Then \cite{roditty2004} show that this process correctly maintains the clusters. We then repeat this process for all the $k+1/\rho+1$ iterations. We add a hopset edge between each $z \in A_i \setminus A_{i+1}$ and all nodes $v \in C(z)$ and set the weight of this edge to $w(v,z)=d_G(v,z)$.


\subsection{Proof of Lemma \ref{lem:dec_bound_cluster}}
Next, we argue that using the above algorithm (with the modified probabilities), we can get the following extension of a result by \cite{roditty2004} that is crucial in bounding the total update time throughout this paper:
 \begin{lemma}
For every $v \in V$ and $0 \leq i \leq k-1$, the expected total number of times the edges incident on $v$ are scanned over all trees for each $w \in A_i$ (i.e.~trees on $C(w)$) is $O(\hat{d}/q_i)$, where $q_i$ is the sub-sampling probability.
 \end{lemma}
 \begin{proof}
 Let $w \in A_i \setminus A_{i+1}$. The edges of a node $v \in V$ is scanned when $v$ joins $C(w)$, and any time $d(v,w)$ is increased until $v$ leaves $C(w)$. 
 We start by analyzing the total cost of joining new clusters. Recall that $C(w)=\{ v \in V: d(v,w) <d(w, A_{i+1}) \}$. Since we are in a decremental setting, $v$ can join $C(w)$ only when $d(w, A_{i+1})$ increases, and this can happen at most $\hat{d}$ times \textit{per tree}. As in the static setting, at any time, $v$ joins at most $\tilde{O}(1/q_i)$ trees, since the number of clusters $v$ belongs to is dominated by a geometric random variable with parameter $q_i$. We will use a similar argument for analyzing the total number of clusters each node belongs to over time.
 Hence the total time for nodes joining new clusters is $\tilde{O}(\hat{d}m/q_i)$. 
 Next, we consider the case when after the deletion the distance between $v$ and the center increases. This will let us bound the number of times the edges incident on $v$ are scanned for a tree rooted at some node in $A_i$.
 Let $d_t(w,v)$ denote the distance between $v$ and $w$ at time $t$ (after $t$ deletions), and let $C_t(w)$ denote the cluster rooted at $w$ at time $t$. We bound the number of indices $t$ for which $v \in C_t(w)$ and $d_t(w,v) < d_{t+1}(w,v)$.
Let $w_{t,1}, w_{t,2},...$ be the sequence of nodes in $A_i$ sorted based on their distance from $v$ at time $t$. Ties will be broken by ordering based on pairs $(d_t(v,w),d_{t+1}(v,w))$, i.e.~nodes with the same distance from $v$ at time $t$ will be sorted based on their distance at time $t+1$. This ensures that if $d_t(v,w_{t,j}) < d_{t+1}(v,w_{t,j})$, then $d_t(v,w_{t,j}) < d_{t+1}(v,w_{t+1,j})$. Same as before $\Pr[v \in C_t(w_{t,j})] \leq (1-q_i)^{j-1}$, since $v \in C_t(w_t,j)$ only if for all $j' <j$ we have $w_{t,j'} \in A_i \setminus A_{i+1}$. Let $I=\{(t,j) \mid d_t(v,w_{t,j}) < d_{t+1}(v,w_{t,j}) \leq \hat{d}\}$. Then since edges incident to $v$ are scanned only if their distance increases, the expected number of times they are scanned over all trees rooted at centers in $A_i$ is at most $\sum_{(t,j)} \Pr[v \in C_t(w_{t,j})]$. Also, by definition for a fixed $j$ there can be at most $\hat{d}$ pairs of form $(t,j)$. In other words, the distance to the $j$-th closest vertex can increase at most $\hat{d}$ times, and hence,
\[\sum_{(t,j)} \Pr[v \in C_t(w_{t,j})] \leq \hat{d }\sum_{j \geq 1} (1-q_i)^{j-1} \leq \hat{d}/q_i.\]  
\end{proof}

\section{Monotone ES tree}\label{app:monotone_es}
In this section, we explain the monotone ES tree idea and how it can be used for maintaining single-source shortest path up to a given depth $D$. Using the monotone ES tree ideas may impact the stretch, and clearly do not apply to all types of insertions but only for insertion of certain structural properties.  In Section \ref{sec:ss_stretch}, we will prove that specifically for the insertions in our restricted hopset algorithm the stretch guarantee holds.
We show how to handle edge insertions by using a variant of the monotone ES-tree algorithm~\cite{henzinger2014} (and further used in the hopset construction of \cite{henzinger2016}). This algorithm is given as Algorithm~\ref{alg:estree}. 
The idea in a monotone ES tree is that if an insertion of an edge $(u,v)$ causes the level of a node $v$ to decrease, we will not decrease the level. In this case we say the edge $(u,v)$ and the node $v$ are \textit{stretched}. More formally, a node $v$ is stretched when $L(v) > \min_{(x,v) \in E} \dest(x) + w(x, v)$.

We observe multiple properties of the monotone ES tree algorithm as observed by \cite{henzinger2014, henzinger2016} that will be helpful in analyzing the stretch later:
\begin{itemize}
    \item The level of a node never decreases.
    \item Only an inserted edge can be stretched.
    \item While an edge is stretched, its level remains the same. In other words, a stretched edge is not going to get stretched again unless it is deleted (or get a distance increase).
\end{itemize}

Also observe that we never underestimate the distances. This is clearly true for any edge weights obtained by the rounding in Lemma \ref{lem:rounding}. It is also easy to see this is true for the stretched edges for the following reason: For any node $v$, the algorithm maintains the invariant that $L(s,v) \geq min_{(x,v) \in E}L(s,x)+w(x,v)$. In other words, $L(s,v)$ is either an estimate based on rounding that is at least $d_G(s,v)$ or it is larger than such an estimate.

\begin{algorithm} [h]
\SetKwProg{Fn}{Function}{}{}

\Fn{\textsc{Init}$(G,s, D)$}{

  $E:=E(G)\cup \{e_v=(s,v):v\in V(G)\setminus\{s\},w(e_v)=D+1\}$ \tcc{This ensures that distances are maintained up to level $D$}
  
  \For{$v\in V$}{
    $L(s,v):=0$
  }
  
  \For{$v\in V$}{
    \textsc{Update}($T(s),v$)
  }
}

  \vspace{2mm}
  
    


\Fn{\textsc{InsertEdge}$(T(s),(a,b),c)$}{ \tcc{Insert an edge in the tree rooted at $s$}
   $E := E \cup \{(a,b)\}$
   
   $w(a,b):={c}$
   
   \textsc{Update}($T(s),b$)
}

\Fn{\textsc{Update}$(T(s),v)$}{
  $upd := \min_{(x,v) \in E} L(s,x) + w(x, v)$
  
  \If{$v=s$ \textbf{or} $L(s,v) \geq upd$}{ \tcc{Node $v$ is stretched.}
  \Return}
  $L(s,v):= upd$
  
  \For{$(v,y)\in E(G)$}{
    \textsc{Update}($T(s),y$)
  }
}
\caption{\label{alg:estree} Maintaining a monotone ES tree up to depth $D$ on $G$. Note that edge deletion can be achieved by setting the edge weight to $\infty$.}
\end{algorithm}

\end{document}